\documentclass[5p,final,number]{elsarticle}
\usepackage{graphicx}
\usepackage{subcaption}
\usepackage{amsmath,amsthm,amssymb}
\newtheorem{theorem}{Theorem}[section]
\newtheorem{lemma}[theorem]{Lemma}
\newtheorem{corollary}[theorem]{Corollary}
\newtheorem{proposition}[theorem]{Proposition}
\newtheorem{remark}[theorem]{Remark}
\newtheorem{assumption}[theorem]{Assumption}

\newcommand{\bs}[1]{\boldsymbol{#1}}
\def\N{\mathbb{N}}
\def\Q{\mathbb{Q}}
\def\R{\mathbb{R}}

\usepackage{color}
\definecolor{dgreen}{rgb}{0,0.6,0}
\definecolor{lgray}{rgb}{0.7,0.7,0.7}

\def\oL{\bs{l}}
\def\cK{\bs{k}}
\def\cKT{\bs{k}^T}
\def\ctKT{\tilde{\bs{k}}^T}
\def\cv{\bs{c}^{T}}
\newcommand{\cs}[1]{c_{#1}}

\def\cssym{\mathfrak{c}}
\def\clsym{\mathfrak{C}}
\newcommand{\lamBar}{\bar{\lambda}}

\newcommand{\lamAa}{\lambda_{1,k}^{a}}
\newcommand{\lamBa}{\lambda_{2,k}^{a}}
\newcommand{\lambAa}{\bar{\lambda}_{1,k}^{a}}
\newcommand{\lambBa}{\bar{\lambda}_{2,k}^{a}}

\graphicspath{{.},{./figures/}}
\begin{document}
\begin{frontmatter}
  \title{Finite-dimensional output stabilization for a class of linear
      distributed parameter systems --- a small-gain
    approach\tnoteref{t1}}
  \tnotetext[t1]{Partially funded by the Deutsche
    Forschungsgemeinschaft (DFG, German Research Foundation) –
    Project-IDs 274853298 (Grüne) and 274852737 (Meurer).
    We thank Alexander Schaum for stimulating discussions during
    L.~Gr\"une's visit to Kiel in March 2020 and Andrii Mironchenko
    for pointing us to the literature on small-gain approaches for
    infinite-dimensional systems.}
  \author[lg]{L. Gr\"une\corref{cor1}}
  \ead{lars.gruene@uni-bayreuth.de}
  \cortext[cor1]{Corresponding author}
  \address[lg]{Chair of Applied Mathematics, Mathematical Institute, University of Bayreuth, 95440 Bayreuth, Germany}
  \author[tm]{T. Meurer}
  \ead{tm@tf.uni-kiel.de}
  \address[tm]{Chair of Automation and Control,
    Kiel University, Kaiserstrasse 2, 24143 Kiel,
    Germany}
  \begin{abstract}
    A small-gain approach is proposed to analyze closed-loop
      stability of linear diffusion-reaction systems under
      finite-dimensional observer-based state feedback control. For
      this, the decomposition of the infinite-dimensional system
      into a finite-dimensional slow subsystem used for design
      and an infinite-dimensional residual fast subsystem is considered. The
      effect of observer spillover in terms of a particular (dynamic)
      interconnection of the subsystems is thoroughly analyzed for
      in-domain and boundary control as well as sensing. This leads to
      the application of a small-gain theorem for interconnected
      systems based on input-to-output stability and unbounded
      observability properties. Moreover, an approach is presented for
      the computation of the required dimension of the slow subsystem
      used for controller design. Simulation scenarios for both
      scalar and coupled linear diffusion-reaction systems are used to
      underline the theoretical assessment and to give insight into
      the resulting properties of the interconnected systems.
  \end{abstract}
\end{frontmatter}
\section{Introduction}
Spillover is an inherent performance and stability issue when
addressing the control of distributed parameter systems based on
finite-dimensional approximations. The term spillover was
characterized in, e.g., \cite{balas:78c,balas:83,guelich:84} and refers to
deterioration of the control performance due to the
infinite-dimensional residual dynamics that is neglected during
control design when taking into account approximation schemes such as
modal, Galerkin or weighted residuals methods
\cite{balas:86,curtain:82,curtain:85,curtain:86}. In particular the
so-called observation spillover might be a source of instability
of the closed-loop control system. Observation spillover can arise
when applying the combination of state feedback controller and state observer -- both
designed based on the finite-dimensional approximation -- to the
original distributed parameter system due to the additional feedback
loop generated by the injection of the contribution of the residual
dynamics to the system output into the observer. 

Finite-dimensional compensator design for distributed parameter
systems (DPSs) has a long history with contributions from different
authors, e.g.,
\cite{gilles:73,balas:78b,balas:83,curtain:82,curtain:86,curtain_robust:86,schuhmacher:83}. The
  particular combined controller and observer structure used in this paper
  seems to be used first in \cite{balas:79} and later in, e.g., 
  \cite{sakawa:FeedbackStabilizationLinear:1983}. Explicit 
formulas to determine the effect of a finite-dimensional modal controller on the
original infinite-dimensional system are derived in \cite{curtain:85}. For this degenerate operator
perturbations are studied but without providing a criterion concerning
the order of the (modal) subsystem to design the
finite-dimensional compensator and observer. Related results are
provided in \cite{curtain:ControllerDesignDistributed:1986} based on
Hankel-norm approximation. The connection between spill\-over and
robustness is analyzed in 
\cite{bontsema:NoteSpilloverRobustness:1988}. A Lyapunov-based
stability analysis of the closed-loop control system with
finite-dimensional modal controller is presented in, e.g.,
\cite{hagen:FinitedimensionalDecentralizedControl:2001}.  
To reduce spillover effect when controlling distributed parameter
systems different measured have been suggested. These include residual mode filters
\cite{balas:FinitedimensionalControllersLinear:1988}, augmented observers
\cite{chait:ControlDistributedParameter:1987} or so-called
cascaded output observers
\cite{harkort:FinitedimensionalObserverbasedControl:2011}. The
  latter reference also considers the a priori determination of the order of the
  stabilizing compensator for systems with bounded input and output
  operator while in general
the necessary order of the reduced system is not specified explicitly
but should be chosen sufficiently large without providing a
computational criteria. 

Feedback stabilization based on reduced order models for large scale
(converged) approximations of linear and nonlinear distributed
parameter systems are suggested in, e.g.,
\cite{shvartsman:NonlinearModelReduction:1998,cao:NumericalCriterionStabilization:2001}. In
\cite{cao:NumericalCriterionStabilization:2001} numerical tools are
used to determine a lower bound on the order of the reduced system so
that the stabilization of a steady state is ensured. Lyapunov theory
and modal decomposition are applied, e.g., in \cite{coron_trelat:04} for a
semilinear 1D heat equation or in
\cite{prieur:FeedbackStabilization1D:2019} for a 1D linear heat with
input delay. Herein a separation between the
finite-dimensional and the infinite-dimensional residual dynamics is
considered for the stability analysis by assuming direct availability
of the modal states without amending the control loop by an
observer. These results are extended in
\cite{katz:ConstructiveMethodFinitedimensional:2020} by developing a
finite-dimensional observer-based control for a 1D heat equation which
relies on Lyapunov's stability theory and linear matrix inequalities
to formulate conditions for the determination of the dimension of the
reduced order system. Delayed input and output are addressed in
\cite{katz:DelayedFinitedimensionalObserverbased:2021} for a scalar
diffusion-reaction equation. Related results are proposed in
\cite{lhachemi:FinitedimensionalObserverbasedPI:2020} for observer-based
PI-control and in \cite{lhachemi:LocalOutputFeedback:2021} taking into
account saturated control.

Differing from previous work this contribution makes use of a
small-gain theorem to assess closed-loop stability of the
interconnection between a finite-dimensional state
feedback control with observer using modal approximation and the
infinite-dimensional residual system. This enables us to verify that
the stabilization of a suitable low-order subsystem ensures stability
of the infinite-dimensional system under this feedback control and to
compute a lower bound on the order of this subsystem. Here, the classical
decomposition into slow and fast dynamics is exploited and an
observer-based state feedback control is designed for the slow
subsystem. Observer spillover arises as the sensor signal contains
information of both slow and fast dynamics, which induces additional
feedback loops that are not considered during the design. Based on
the eigenvalue distribution of the system operator and certain
characteristic features of the input and output operators a sequence
of estimates for the fast (residual) dynamics under observer-based
state feedback control is determined addressing both in-domain and
boundary actuation and sensing. The preliminary results lead to the
application of a small-gain theorem for interconnected systems based
on input-to-output stability and unbounded observability
properties. To address the dimension of the slow subsystem used for
controller design, a numerical approach is presented and illustrated
in simulation scenarios for both scalar and coupled linear
diffusion-reaction systems.

The paper is organized as follows. A prototype example is
introduced in Section \ref{sec:prototype} to motivate the
formulation of an abstract model in Section
\ref{sec:abstract_model} and the decomposition into slow and fast
dynamics as well as observer-based state feedback control
design. Based on this, auxiliary estimates and results are provided
in Section \ref{sec:assumptions} to prepare the main small-gain
result in Section \ref{sec:main} to confirm closed-loop
stability. Section \ref{sec:comp_n} summarizes a computational
approach to determine the minimal order of the slow subsystem used
for control design. Simulation results for scalar and coupled linear
diffusion-reaction systems in Section \ref{sec:sim} are presented to
confirm the theoretical assessment. Some final remarks conclude the
paper.
\subsection*{Notation}
Given vectors
  $\bs{x}_s(t)=[x_1(t),\ldots,x_n(t)]^{T}\in\mathbb{R}^{n}$,
  $\bs{e}_s(t)=[e_1(t),\ldots,e_n(t)]^{T}\in\mathbb{R}^{n}$, $\bs{x}_f(t)=[x_{n+1}(t),x_{n+2}(t),\ldots]^{T}\in\mathbb{R}^{\infty}$ we use the following norms:
\[ \|\bs{x}_s\|_2 := \sqrt{\sum_{k=1}^n x_k^2}, \quad \|\bs{e}_s\|_2 := \sqrt{\sum_{k=1}^n e_k^2}, \]
\[  \|\bs{x}_s\|_{2,\infty} := \sup_{t\ge 0} \|\bs{x}_s(t)\|_2 \quad \|\bs{e}_s\|_{2,\infty} := \sup_{t\ge 0} \|\bs{e}_s(t)\|_2,\]
\[  \|(\bs{x}_s)\vert_{[t_1,t_2]}\|_{p,\infty} := \sup_{\tau\in [t_1,t_2]}
    \|\bs{x}_s(\tau)\|_p,\quad p\in\{1,2\},
\]
\[ \|\bs{x}_f\|_1 := \sum_{k\ge n+1} |x_k|, \quad \|\bs{x}_f\|_2 := \sqrt{\sum_{k\ge n+1} x_k^2}, \]
\[  \|\bs{x}_f\|_{p,\infty} := \sup_{t\ge 0} \|\bs{x}_f(t)\|_{p},
    \quad p\in\{1,2\}.
\]
Moreover, for matrices and linear operators we use the usual induced
operator norms. 
\section{A prototype system}\label{sec:prototype}
We motivate the study in this paper by considering the (unstable) linear diffusion--reaction system
\begin{subequations}
  \label{eq:dr:sys}
  \begin{align}
    &\partial_t x = \partial_z^2 x + r x + b u_1,&&z\in(0,1),~t>0\\
    &\partial_z x\vert_{z=0}=0,\quad x\vert_{z=1} = u_2,&&t>0\\
    &x\vert_{t=0}=x_0,&&z\in[0,1]. 
  \end{align}
\end{subequations}
Let $X=L^2(0,1)$ denote the state space and introduce the self-adjoint
operator $A x = \partial_z^2x$ with domain $D(A) = \{x\in
H^2(0,1)\vert\,\partial_z x(0)=x(1)=0\}$. 
The eigenproblem for $A$ reads $A\phi=\mu\phi$, $\phi\in D(A)$.  
Its solution can be obtained by directly solving the differential
equation and taking into account the boundary conditions. In
particular it follows that $\phi_k = \sqrt{2} \cos(\omega_k z)$,
$\mu_k = -(\omega_k)^2$, $k\in\mathbb{N}$ for $\omega_k=\frac{2k-1}{2}\pi$. 
The sequence $(\phi_k)_{k\in\mathbb{N}}$ form an orthonormal Riesz basis for
$L^2(0,1)$. Taking into account either operator extensions
\cite[Section 13.7]{TucW09}, Green's theorem or modal transformation the system
\eqref{eq:dr:sys} can be projected onto the basis
$(\phi_k)_{k\in\mathbb{N}}$ even taking into account the inhomogeneous
boundary condition at $z=1$. Let $\langle\cdot,\cdot\rangle$ denote
the inner product in $L^2(0,1)$, then
\begin{align*}
  \langle \partial_t x,\phi_k\rangle &= \langle
                                       \partial_z^2 x,\phi_k\rangle + r
                                       \langle x,\phi_k\rangle
                                       +\langle b,\phi_k\rangle u_1.
\end{align*}
Interchanging time differentiation and integration and using
integration by parts taking into account the boundary conditions provides
\begin{align*}
  &\partial_t\langle x,\phi_k\rangle\\
  &=
    -\partial_z\phi_k\vert_{z=1}u_2 + 
    \langle
    x,\partial_z^2\phi_k\rangle + r
    \langle x,\phi_k\rangle +\langle b,\phi_k\rangle u_1 \\
  & = -\partial_z\phi_k\vert_{z=1} u_2 + 
    \mu_k \langle
    x,\phi_k\rangle + r
    \langle x,\phi_k\rangle +\langle b,\phi_k\rangle u_1 .
\end{align*}
Denoting $x_k=\langle x,\phi_k\rangle$, $b_ {1,k}=\langle b,\phi_k\rangle$ and $b_{2,k} =
-\partial_z\phi_k\vert_{z=1}$ the latter equation can be
re-written as the infinite-dimen\-sion\-al system of ODEs in diagonal form
\begin{subequations}
  \label{eq:dr:spectral}
  \begin{align}
    \dot{x}_k &= (r+\mu_k)x_k +b_{1,k} u_1 +b_{2,k} u_2,&&k\in\mathbb{N}\\
    x_k(0) &= \langle x_0,\phi_k\rangle = x_k^0. 
  \end{align}
\end{subequations}
Let subsequently $\lambda_k:=r+\mu_k$. In view of the coefficient
$b_{l,k}$ the prototype problem \eqref{eq:dr:spectral} involves both
in-domain $(l=1)$ and boundary control $(l=2)$, respectively.
\section{An abstract model}\label{sec:abstract_model}
In this section we {specify} an abstract model that captures the
properties of the prototype system just discussed. The results in the
paper will be formulated for this abstract model. All necessary
assumptions will be summarised in the next section.
\subsection{Problem setup}
  We consider systems given by the abstract Cauchy problem 
  \begin{subequations}
    \label{eq:dr:abstract}
    \begin{align}
      \dot{\bs{x}} &= A\bs{x} +
                     B\bs{u},\quad t>0,\quad\bs{x}(0)=\bs{x}_0\in D(A)\\
      \bs{y} &= C\bs{x},\quad t\geq 0
    \end{align}
  \end{subequations}
  on the Hilbert space $X$ equipped with the inner product $\langle
  \cdot,\cdot \rangle_X$. The system operator is denoted by $A$ with domain
  is denoted by $D(A)\subset X$.
  \begin{assumption}\label{assump:abstract_system}
    The abstract system fulfills the following assumptions:
    \begin{enumerate}
    \item\label{assump:abstract_system:op} The operator $A$ is a Riesz spectral operator in the
      sense of \cite[Section 2.3]{curtain_zwart:95}, i.e., $A$ has only discrete 
      eigenvalues $\lambda_k$ with $\sup_{k\in\mathbb{N}}{\rm Re}\{ \lambda_k\}<\infty$ and the eigenvectors $\bs{\phi}_k\in D(A)$ and
      the adjoint eigenvectors $\bs{\psi}_k\in D(A^{\ast})$ form orthonormal
      Riesz bases so that $\langle\bs{\phi}_k,\bs{\psi}_l\rangle_X =
      \delta_{k,l}$ with $\delta_{k,l}$ denoting the Kronecker delta.
    \item\label{assump:abstract_system:eig} The eigenvalues
      $\lambda_k$ of $A$ are real-valued with {$\lambda_k\to -\infty$ as $k\to\infty$ and are} arranged so
      that $\lambda_{k+1}\leq \lambda_k$ for all $k$. 
    \item\label{assump:abstract_system:io} The possibly unbounded operators $B$ and $C$ are assumed to
      be admissible control and output operators, respectively, in the
      sense of \cite{ho:83} or \cite[{Chapter 4}]{TucW09} {with finite dimensional control space $U = \R^m$ and output space $Y = R^l$.} %
    \end{enumerate}
  \end{assumption}
  Assumption \ref{assump:abstract_system:op} implies that the state
  $\bs{x}\in X$ can be represented by the Fourier series $\bs{x} =
  \sum_{k=1}^{\infty} x_k\bs{\phi}_k$ and the operator $A:D(A)\to X$ admits the
  decomposition
  \begin{align}
    \label{eq:abstract:op:series}
    A\bs{x} %
              &= \sum_{k=1}^{\infty}
    \lambda_k x_k\bs{\phi}_k\quad\forall \bs{x}\in D(A)\\
    D(A) &= \bigg\{\bs{x}\in X:~\sum_{k=1}^{\infty} (1+\vert
           \lambda_k\vert^2)\vert x_k\vert^2<\infty\bigg\},
  \end{align}
  where $x_k:=\langle \bs{x},\bs{\psi}_k\rangle_X$ 
  represents the $k$th Fourier or modal coefficient \cite{gohberg_krein:69,curtain_zwart:95,young:01,guo:ControlWaveBeam:2019}. The operator $A$
  is also called diagonalizable \cite[{Section 2.6}]{TucW09}. By \ref{assump:abstract_system:op} it also
  follows that the adjoint operator is diagonalizable with eigenvalues
  $\overline{\lambda_k}=\lambda_k$ noting Assumption
  \ref{assump:abstract_system} \ref{assump:abstract_system:eig}. 
  Proceeding as in \cite[{Section 2.10}]{TucW09}, let $X_1$ denote the space $D(A)$ 
  equipped with the norm $\|\bs{x}\|_1=\|(\beta I-A)\bs{x}\|_X$ for
  some $\beta\in\rho(A)\ne\emptyset$. Note that the norms generated for
  different $\beta$ are equivalent in the graph norm so that
  $\|\bs{x}\|_1$ is independent of the particular choice of
  $\beta$. Let $X_{-1}$ denote the dual of $X_1$ with respect to the
  pivot space $X$, i.e. $X_1\subset X \subset X_{-1}$ with continuous
  dense injections. Similar to the Fourier representation of $\bs{x}$ via the sequence $(x_k)_k$, any linear operator $\bs{z}$ can be represented by the sequence $(z_k)_k$ given by $z_k := \bs{z}\phi_k$ for $k\in\N$. With this representation,  
  the space $X_{-1}$ can be identified with the
  space of sequences $z=(z_k)_k$ for which
  \begin{align}
    \label{eq:X1_dual:norm}
    \|z\|_{-1}^2 = \sum_{k=1}^{\infty}\frac{\vert z_k\vert^2}{1+\vert\lambda_k\vert^2}<\infty.
  \end{align}%
  In a similar fashion, the Riesz basis property and $A,\,A^{\ast}$ being diagonalizable
  imply, see, e.g.,
  \cite{rebarber:NecessaryConditionsExact:2000,weiss:EigenvaluesEigenvectorsSemigroup:2011}  that any input
  operator $B\in\mathcal{L}(U,X_{-1})$ can be 
  represented by %
  a sequence in $U$ according to
  \begin{align}\label{eq:Bu}
    B\bs{u} &= \sum_{k=1}^{\infty} \langle \bs{b}_k,\bs{u}\rangle_U
              \bs{\phi}_k,\quad \bs{b}_k=B^{\ast}\bs{\psi}_k .
  \end{align}
  If we define the sequence $v=(\langle \bs{b}_k,\bs{u}\rangle_U)_k$, then
  $\|v\|_{-1}<\infty$. Taking into account \ref{assump:abstract_system:eig} and
  \ref{assump:abstract_system:io} it can be shown, see, e.g.,
  \cite{weiss:EigenvaluesEigenvectorsSemigroup:2011}, that there exists
  $m_b>0$ so that $\|\bs{b}_k\|\leq m_b(1+\vert\lambda_k\vert)$ for
  all $k\in\mathbb{N}$. We also note that the input operator $B$ is called admissible, if
  \eqref{eq:dr:abstract} considered as an abstract Cauchy problem with
  values in $X_{-1}$ has a continuous $X$-valued {mild} solution for any
  $\bs{u}\in L^2([0,\infty);U)$ \cite{ho:83}, \cite[{Definition 4.1.5}]{TucW09}. {Throughout this paper, we consider these mild solutions.}
  {We} refer to \cite{ho:83} and \cite[{Chapter 10}]{TucW09} for the
  formulation of boundary control problems in the form
  \eqref{eq:dr:abstract} with unbounded input operator $B$ using so-called operator extensions. 

  Let $C_j$ denote the $j$-th component of the output operator $C$. Then, using 
  $c_{j,k} = C_j\phi_k$, the identity 
  \begin{align}
    \label{eq:abstract:C_j}
    C_j\bs{x} = \sum_{k=1}^\infty
    c_{j,k} x_k
  \end{align}
  holds provided the infinite sum is absolutely convergent. Due the fact that the admissibility assumption \ref{assump:abstract_system:io} demands that $C_j\in\mathcal{L}(X_1,\mathbb{R})$, by \eqref{eq:X1_dual:norm} this is in particular the case if $\bs{x}\in X_1$.
Theorem 5.3.2 from \cite{TucW09} and the eigenvalue condition in \ref{assump:abstract_system:eig} imply that the admissibility condition in 
\ref{assump:abstract_system:io} for $C$ is equivalent to the existence of a constant 
$m_c\geq 0$ such that
  \begin{align}
    \label{eq:abstract:c_jk:condition}
    \frac{1}{h}\sum_{\lambda_k \ge -h} \vert c_{j,k}\vert^2\leq m_c,
  \end{align}
  for all $h>0$. We note that in general this is a more demanding condition than \eqref{eq:X1_dual:norm}.
\begin{remark}
  In order to simplify the exposition we restrict ourselves to
  ~\eqref{eq:dr:abstract}
  in the SISO-case, i.e., $B\bs{u}=\bs{b}u$,
  $\bs{y}=C\bs{x}=y$, and explain the necessary changes for the MIMO-case in
  Remark \ref{rem:mimo}.
\end{remark}
\subsection{System decomposition and finite-dimensional observer-based
  control design}
In view of the orthonormality property of the eigenvectors and adjoint
eigenvectors the orthogonal projections $P_s\bs{x} =
\sum_{k=1}^{n}x_k\bs{\phi}_k$ and $P_f\bs{x} =
\sum_{k=n+1}^{\infty}x_k\bs{\phi}_k$ so that
$\bs{x}=P_s\bs{x}+P_f\bs{x}$ can be introduced. By Assumption \ref{assump:abstract_system:eig} and
making use of $P_s$ and $P_f$  the system dynamics
\eqref{eq:dr:abstract}, respectively, can be split into a
finite-dimensional slow and an infinite-dimensional fast dynamics.
In particular, if we identify $\bs{x}_s = P_s\bs{x}$ with $[x_1,x_2,\ldots,x_n]^T$ and
$\bs{x}_f=P_f\bs{x}$ with $[x_{n+1},x_{n+2},\ldots]^T$, then
\eqref{eq:dr:spectral} can be written in the form 
\begin{subequations}
  \label{eq:slowfastsys}
  \begin{align}
    \dot{\bs{x}}_s &= \Lambda_s \bs{x}_s + \bs{b}_s u\label{eq:slowsys}\\
    \dot{\bs{x}}_f &= \Lambda_f \bs{x}_f + \bs{b}_f u\label{eq:fastsys}
  \end{align}
\end{subequations}
with $\Lambda_s=\text{diag}\{\lambda_1,\ldots,\lambda_n\}$,
$\Lambda_f=\text{diag}\{\lambda_{n+1},\lambda_{n+2},\ldots\}$,
$\bs{b}_s=[b_1,\ldots,b_n]^T$, and
$\bs{b}_f=[b_{n+1},b_{n+2},\ldots]^T$, where $b_k = \langle \bs{b},\bs{\psi}_k\rangle_X$. 
We assume that $\lambda_{n+1}<0$, which for our model problem can always be achieved
if $n$ is sufficiently large. 
The output of the system can be equivalently split into
\begin{align}
y = Cx = \cv_s \bs{x}_s + C \bs{x}_f,
\end{align}
where $\cv_s=[\cs{1},\ldots,\cs{n}]$ with $c_k = C\phi_k$. If the infinite sum in \eqref{eq:abstract:C_j} is absolutely convergent, then we can write 
\begin{align} \label{eq:fast_output_sum}
C \bs{x}_f = \sum_{k=n+1}^{\infty}C\bs{\phi_k}x_k = \cv_f \bs{x}_f \end{align}
with $\cv_f=[\cs{n+1},\cs{n+1},\ldots]$ with $c_k = C\phi_k$. Whenever we use the representation \eqref{eq:fast_output_sum} in this paper, we will check that the absolute convergence property holds (cf.\ \eqref{eq:Cf_times_x:z1} and Lemma \ref{lem:ios_estimate:z:case:1} as well as  \eqref{eq:Cf_times_x:z} and Lemma \ref{lemma:znormest}).

Consider now the finite-dimensional (slow) system
\begin{subequations}
  \label{eq:findim:slow}
  \begin{align}
  \dot{\bs{x}}_s &= \Lambda_s \bs{x}_s + \bs{b}_s u\\
  {y}_s &= \cv_s \bs{x}_s
  \end{align}
\end{subequations}
Assuming stabilizability and detectability, we can design a stabilizing dynamic output feedback law based on a Luenberger observer. The observer is of the form
\begin{subequations}
  \begin{align}
  \dot{\hat{\bs{x}}}_s &= \Lambda_s \hat{\bs{x}}_s + \bs{b}_s u + \oL({y}_s-\hat{{y}}_s)\label{eq:obs}\\
 \hat{{y}}_s &= \cv_s \hat{\bs{x}}_s
  \end{align}
\end{subequations}
and the resulting control reads $u=-\cKT\hat{\bs{x}}_s$.
  Defining $\bs{e}_s:= \bs{x}_s - \hat{\bs{x}}_s$ we can rewrite this as 
\begin{subequations}
  \begin{align}
  \dot{\bs{e}}_s &= (\Lambda_s-\oL \cv_s) \bs{e}_s\label{eq:err}\\
  u &= -\cKT(\bs{x}_s - \bs{e}_s)
  \end{align}
\end{subequations}
If we entirely neglect the fast, infinite-dimensional subsystem \eqref{eq:fastsys}, we end up with the finite-dimensional closed-loop system
\begin{subequations}
  \label{eq:finsys:all}
  \begin{align}
  \dot{\bs{x}}_s &= \Lambda_s \bs{x}_s - \bs{b}_s \cKT(\bs{x}_s - \bs{e}_s)\label{eq:finsys1}\\
  \dot{\bs{e}}_s &= (\Lambda_s-\oL \cv_s) \bs{e}_s\label{eq:finsys2}
  \end{align}
\end{subequations}
which is asymptotically stable if $\cK$ and $\oL$ are appropriately chosen.
We note that while $\cK$ and $\oL$ depend on $n$, we impose in Assumption \ref{assump:general}, below, that their
norm is bounded independent of $n$. This assumption is satisfied,
  e.g., if we only shift a finite number of eigenvalues with these feedback
  laws. We illustrate this for $\cK$. Suppose we want to design
  ${\cK}$ such that $\Lambda_s-\bs{b}_s\cK^T$ has the eigenvalues
  $\kappa_1,\kappa_2,\ldots,\kappa_j,\lambda_{j+1},\lambda_{j+2},\ldots$,
  where $\lambda_i$ are the eigenvalues of $\Lambda_s$. Recalling that
  $\Lambda_s$ is a diagonal matrix, we can then write $\Lambda_s$,
  $\bs{b}_s$ and $\cK$ as
  \[
    \Lambda_s = \left(\begin{array}{ll}\Lambda_1 & 0 \\ 0 & \Lambda_2\end{array}\right), \quad 
    \bs{b}_s = \left(\begin{array}{l} \bs{b}_1 \\
                       \bs{b}_2\end{array}\right), \; \mbox{ and } \;
                   \cK = \left(\begin{array}{l} \cK_1 \\
                                 \cK_2\end{array}\right),
                           \]
with $\Lambda_1\in\R^{j\times j}$, $\bs{b}_1,\cK_1\in\R^j$,
$\Lambda_2\in\R^{(n-j)\times (n-j)}$ and
$\bs{b}_2,\cK_2\in\R^{n-j}$. Then we choose $\cK_1$ such that the
matrix $\Lambda_1-\bs{b}_1\cK_1^T$
has the eigenvalues $\kappa_1,\kappa_2,\ldots,\kappa_j$. Clearly, the entries of $\cK_1$ and thus the norm of this vector do not depend on $n$. Setting $\cKT := (\cKT_1,{\bs 0}^T)^T$ with $\bs{0}\in\R^{n-j}$, the norm of $\cK$ is also independent of $n$. This feedback law yields the desired eigenvalues, since  
\[
  \Lambda_s-\bs{b}_s\cKT= \left(\begin{array}{rl}\Lambda_1 - \bs{b}_1\cKT_1 & 0 \\ -\bs{b}_2\cKT_1 & \Lambda_2\end{array}\right)\]
is a block lower triangular matrix whose eigenvalues coincide with that of $\Lambda_1 - \bs{b}_1\cKT_1$ and $\Lambda_2$.

In practice, the output ${y}_s=\cv_s\bs{x}_s$ will not be available
for implementation. Rather, ${y} = C\bs{x} = \cv_s \bs{x}_s + {C\bs{x}_f}$ can be measured. This means that ${y}_s$ in
\eqref{eq:obs} is replaced by ${y}={y}_s + {y}_f$ with $y_f={C\bs{x}_f}$. As a
consequence, \eqref{eq:err} becomes \[ \dot{\bs{e}}_s = (\Lambda_s-\oL \cv_s) \bs{e}_s - \oL {y}_f = (\Lambda_s-\oL \cv_s) \bs{e}_s - \oL {C\bs{x}_f}\]
and the true closed-loop system is described by
\begin{subequations}
  \label{eq:cl:all}
  \begin{align}
    \dot{\bs{x}}_s &= \Lambda_s \bs{x}_s - \bs{b}_s \cKT(\bs{x}_s - \bs{e}_s)\label{eq:cl1}\\
    \dot{\bs{e}}_s &= (\Lambda_s-\oL \cv_s) \bs{e}_s- \oL {C\bs{x}_f}\label{eq:cl2},
                     \intertext{which needs to be completed with}
                     \dot{\bs{x}}_f &= \Lambda_f \bs{x}_f - \bs{b}_f \cKT(\bs{x}_s - \bs{e}_s)\label{eq:cl3},
  \end{align}
\end{subequations}
i.e., with the fast subsystem \eqref{eq:fastsys} for $u=-\cKT(\bs{x}_s
- \bs{e}_s)$. Note that \eqref{eq:cl1} and \eqref{eq:cl2} are
finite-dimensional, while \eqref{eq:cl3} is infinite-dimensional.

If we abbreviate $\tilde{\bs{x}}_s = [\bs{x}_s^T,\bs{e}_s^T]^T$, then
we can write \eqref{eq:cl1}, \eqref{eq:cl2} as
\begin{equation} \dot{\tilde{\bs{x}}}_s = \widetilde A
  \tilde{\bs{x}}_s +
  \tilde{\oL}{C\bs{x}_f}, \label{eq:c1tilde}
\end{equation}
where $\widetilde A$ is a Hurwitz matrix and $\tilde{\oL}
  = (\begin{smallmatrix}0\\-\oL\end{smallmatrix})$. We pick $\tilde\lambda>0$ and
$\widetilde M>0$ such that $\|e^{\widetilde A t}\| \le \widetilde
Me^{-\tilde\lambda t}$ for all $t\ge 0$.
\begin{proposition}
    Let
    \begin{align}
      \label{eq:prop:independency:n}
      \sum_{k=n+1}^\infty \frac{b_k^2}{\lambda_k^2}<\infty,
    \end{align}
    then the bound $\widetilde{M}$ is independent of the dimension $n$ of slow
    subsystem \eqref{eq:slowsys}. %
  \end{proposition}
We note that \eqref{eq:prop:independency:n} follows from the admissibility of $B$, using that $\|v\|_{-1}<\infty$ for $v$ defined after \eqref{eq:Bu} and that $\lambda_k\to -\infty$ as $k\to\infty$ and $\lambda_k < 0$ for all $k \ge n+1$.

\begin{proof}
    Let $\bs{x}_{s,j} = [x_1,\ldots,x_j]^T$, $\bs{e}_{s,j} = [e_1,\ldots,e_j]^T$ and $\tilde{\bs{x}}_{s,j} = [\bs{x}_{s,j}^T,\bs{e}_{s,j}^T]^T$.
    Then, for each $k>j$ with $j$ from the controller and observer
    construction, we obtain the two equations
    \begin{align*}
      \dot x_{k} &= \lambda_{k} x_{k} + b_{k} \bs{z}_x^T
                   \tilde{\bs{x}}_{s,j}\\
      \dot e_{k} &= \lambda_{k} e_{k} + b_{k} \bs{z}_e^T \tilde{\bs{x}}_{s,j},
    \end{align*}
    where $\bs{z}_x^T$ and $\bs{z}_e^T$ are suitable column
    vectors independent of $n$ determined by $\bs{k}^T$ and $\bs{l}$. By construction and the suitable choice of
      $\kappa_1,\ldots,\kappa_j$ there exist positive constants
      $\tilde\lambda_j,\,\widetilde M_j>0$ with $\tilde\lambda_j +
      \lambda_k<0$ for all $k>j$ so that
      $\|\tilde{\bs{x}}_{s,j}\|_2\leq \widetilde M_j e^{-\tilde\lambda_j
        t}\|\tilde{\bs{x}}_{s,j}(0)\|_2$.
      By the variations of constants formula we obtain
      \[ x_{k}(t)  = e^{\lambda_{k}t} x_{k}(0) + \int_0^t
        e^{\lambda_{k}(t-s)} b_{k} \bs{z}_x^T \tilde{\bs{x}}_{s,j}(s)
        ds.\]
      Using the bound on $\tilde{\bs{x}}_{s,j}$, the Cauchy-Schwarz
      inequality $\vert\bs{z}_x^T \tilde{\bs{x}}_{s,j}(s)\vert\leq
      \|\bs{z}_x\|_2\|\tilde{\bs{x}}_{s,j}(s)\|_2$, and noting
      $e^{\lambda_k t}\leq e^{-\tilde\lambda_j t}$ this implies the
      sequence of estimates
      \begin{align*} |x_{k}(t)| & \le e^{\lambda_{k}t}|x_{k}(0)|\\
                                & \quad + \int_0^t
                                  e^{\lambda_{k}(t-s)} \vert b_{k}\vert \|\bs{z}_x\|_2 \widetilde M_j e^{-\tilde \lambda_j s} \|\tilde{\bs{x}}_{s,j}(0)\|_2 ds\\
                                & = e^{\lambda_{k}t}|x_{k}(0)|\\ &
                                                                   \quad + \vert b_{k}\vert \|\bs{z}_x\|_2\|\tilde{\bs{x}}_{s,j}(0)\|_2 \widetilde M_j e^{\lambda_{k}t}\int_0^t e^{-(\tilde \lambda_j + \lambda_{k}) s} ds\\
                                & \leq e^{\lambda_{k}t}|x_{k}(0)| \\ &
                                                                       \quad
                                                                       +
                                                                       \|\bs{z}_x\|_2\|\tilde{\bs{x}}_{s,j}(0)\|_2\widetilde
                                                                       M_j
                                                                       \vert
                                                                       b_{k}\vert\frac{e^{\lambda_k
                                                                       t}-e^{-\tilde
                                                                       \lambda_jt}}{\lambda_k
                                                                       +
                                                                       \tilde
                                                                       \lambda_j}\\
                                & \leq e^{-\tilde\lambda_{j}t}\bigg(
                                  |x_{k}(0)| +
                                  \|\bs{z}_x\|_2\|\tilde{\bs{x}}_{s,j}(0)\|_2
                                  \widetilde M_j
                                  \frac{\vert b_{k}\vert}{\vert \tilde
                                  \lambda_j +\lambda_{k}\vert} \bigg).
      \end{align*}
      The analogous inequality holds for the components of $\bs{e}_s$
      with $|x_{k}(0)|$ replaced by $|e_{k}(0)|$ and $\|\bs{z}_x\|_2$
      by $\|\bs{z}_e\|_2$. Since $(a+b)^2\le 2a^2+2b^2$ and
      $e^{2\lambda_kt} \le e^{-2\tilde\lambda_jt}$, we obtain for the
      $2$-norm after some intermediate but straightforward
      computations making use of $\|\bs{z}\|^2 = \|\bs{z}_x\|_2^2 +
      \|\bs{z}_e\|_2^2$ that 
      \begin{align*}
        \|\tilde{\bs{x}}_{s,n}\|_2^2 &\leq \|\tilde{\bs{x}}_{s,j}\|_2^2
                                       + \sum_{k=j+1}^{n} \vert x_k\vert^2 + \vert e_k\vert^2\\
        &\leq \widetilde M_j^2 e^{-2\tilde\lambda_j
          t}\|\tilde{\bs{x}}_{s,j}(0)\|_2^2\\
        &\quad + 2
          e^{-2\tilde\lambda_{j}t}\sum_{k=j+1}^{n} |x_{k}(0)|^2 +
          |e_{k}(0)|^2\\
        &\quad + 2
          e^{-2\tilde\lambda_{j}t}\sum_{k=j+1}^{n} \|\bs{z}\|_2^2\|\tilde{\bs{x}}_{s,j}(0)\|_2^2\widetilde M_j^2\bigg(\frac{\vert b_{k}\vert}{\vert \tilde
          \lambda_j +\lambda_{k}\vert}\bigg)^2\\
        & \leq \max\{2,\widetilde M_j^2\} e^{-2\tilde\lambda_j
          t}\|\tilde{\bs{x}}_{s,n}(0)\|_2^2\\
        &\quad +  2\widetilde M_j^2
          e^{-2\tilde\lambda_{j}t}\|\bs{z}\|_2^2\|\tilde{\bs{x}}_{s,n}(0)\|_2^2\sum_{k=j+1}^{n} \bigg(\frac{\vert b_{k}\vert}{\vert \tilde
          \lambda_j +\lambda_{k}\vert}\bigg)^2\\
        &\leq \max\{2,2\widetilde M_j^2\}
          e^{-2\tilde\lambda_jt}\|\tilde{\bs{x}}_{s,n}(0)\|_2^2
          \times\\
        &\quad \Bigg( 1 +  \|\bs{z}\|_2^2\sum_{k=j+1}^{\infty} \bigg(\frac{\vert b_{k}\vert}{\vert \tilde
          \lambda_j +\lambda_{k}\vert}\bigg)^2\Bigg)
      \end{align*}
      Assumption \eqref{eq:prop:independency:n} implies the
      convergence of the sum in the latter term which together with
      $\|\bs{z}\|_2$ being independent of $n$ by construction of
      $\bs{k}^T$ and $\bs{l}$ verifies the claim.
\end{proof}
\section{Assumptions and auxiliary estimates}\label{sec:assumptions}
In this section we formulate the precise assumptions on our abstract
model and provide auxiliary estimates, which we will need for the
proof of the main theorem in the next section. %
\begin{assumption}
  \label{assump:general}
  We impose the following assumptions on
  \eqref{eq:slowfastsys} and
  \eqref{eq:cl:all}.
  \begin{enumerate}
    \setcounter{enumi}{3}
  \item\label{assump:general:a0} The matrices $\Lambda_s-\oL \cv_s$ and
    ${\Lambda_s}-b_s\cK$ are Hurwitz with eigenvalues whose real parts
    are smaller than $\delta$, where $\delta<0$ is independent of $n$,
    and $\|\cK\|$, $\|\oL\|$ and $\widetilde M$ are bounded with bounds independent of
    $n$.
  \item\label{assump:general:a1} The sequence of eigenvalues
    $(\lambda_k)_k$ satisfies $\lambda_{k+1}\leq\lambda_k$ for all
    $k\in\N$, $\lambda_k\to-\infty$ for $k\to\infty$, $\lambda_{n+1}<0$ and 
    \begin{align*}
      \sum_{k=1}^{\infty} \frac{1}{\vert \lambda_k\vert}=M_\lambda<\infty,
    \end{align*}
  \item\label{assump:general:a2} The following assumptions will be used
    alternatively. 
    \begin{enumerate}
    \item\label{assump:general:a2a} There is $\alpha>1$ and $d_1>0$ such
    that $\vert\frac{b_k}{\lambda_k}\vert \le \frac{d_1}{k^\alpha}$
    and there is $\cssym_2>0$ with $|\cs{k}|\le \cssym_2$ for
      all $k\ge n+1$. %
  \item\label{assump:general:a2b} Any of the conditions is fulfilled:\end{enumerate}
    \begin{itemize}
    \item The inequality
      $\sum_{k=n+1}^{\infty}\frac{b_k}{\lambda_k}<\infty$
      holds.
    \item There are constants $\cssym_1,\cssym_2,\cssym_3>0$ with $|b_k|\le \cssym_1k$, $|\cs{k}|\le \cssym_2$, $\lambda_k\le -\cssym_3k^2$ for all $k\ge n$.
    \item For each $m\in\N$ there is $\gamma_{m}>0$ such that $|\lambda_k^{-1}-\lambda_{k+m}^{-1}|\le \gamma_{m}k^{-3}$ for all $k\ge n$.
    \item There exist $\cssym_4,\cssym_5,k_1>0$ and pairwise disjoint sets
      $S_j\subset\N$, $j\in\N$, each with at most $s\in\N$ elements,
      $\N=\bigcup_{j\in\N} S_j$, $\min S_j \ge \cssym_4 j$, $\max S_j \le
      \min S_j+k_1$ and $\vert\sum_{k\in S_j} {\cs{k} b_k}/{\lambda_k} \vert \le {\cssym_5}j^{-2}$.
    \end{itemize}
\end{enumerate}
\end{assumption}

Assumption \ref{assump:general:a0} imposes bounds
on the norms of $\cK$ and $\oL$, which, as discussed after
\eqref{eq:finsys:all}, are satisfied if only finitely
many eigenvalues are shifted. Moreover, Assumption \ref{assump:general:a1}
imposes a restriction on the growth of the eigenvalues. In particular
this condition is fulfilled, if $\lambda_k \sim ck^2$ as is typically the case for
diffusion-reaction problems. 
Assumptions \ref{assump:general:a2a} and \ref{assump:general:a2b} will
be used alternatively. Assumption \ref{assump:general:a2a} in
particular implies that the series $\sum_{k=1}^{\infty}
{b_k}/{\lambda_k}$ converges absolutely and is usually satisfied
for in-domain control. Assumption \ref{assump:general:a2b} can be applied if
$\sum_{k=1}^{\infty} \frac{b_k}{\lambda_k}$ does not converge
absolutely, which happens in case of boundary control.
\begin{remark}\label{rem:assump:prototype}
 Assumption \ref{assump:general:a2b} is satisfied for the prototype system from
  Section \ref{sec:prototype} in the boundary control case. This can be checked using 
  $\lambda_k=r-((2k-1)\pi/s)^2$ defined after (\ref{eq:dr:spectral}b)
  together with the $b_{2,k}={(-1)^{k}(2k-1)\pi}/{\sqrt{2}}$ defined
  before (\ref{eq:dr:spectral}a) and the coefficients $\cs{k}$ obtained from point
  measurements at some point $\xi \in(0,1)\cap \Q$. We can then define
  the $S_j$ to be of the form \[ S_j = \{ mk_1+k, (m+1)k_1-k+1\}\]
  with $m=[2j/k_1]$ and $k=j-mk_1/2$. Here $k_1\in\N$ chosen such that
  $\cs{mk_1+k} = \cs{(m+1)k_1-k+1}$ for all these $m$ and $k$. Such
  a $k_1$ exists in case of a point measurements at some point $\xi
  \in(0,1)\cap \Q$ due to the periodicity of the cosine function and
  $[2j/k_1]$ denotes the integer part of $2j/k_1$. 
For $\xi\in(0,1)\setminus\mathbb{Q}$, the existence of $k_1$ with $c_{mk_1+k} = c_{(m+1)k_1-k+1}$, which is crucial for ensuring the last item of Assumption \ref{assump:general:a2b}, cannot be guaranteed. Numerical tests have, however, revealed that one can still find $k_1$ such that this identity is satisfied approximately with a small error. We thus expect that Assumption \ref{assump:general:a2b} is also satisfied for $\xi\not\in\mathbb{Q}$, although a formal proof of this property is beyond the scope of this paper.
\end{remark}

In the remainder of this section we derive auxiliary estimates for the
solutions of the closed loop system \eqref{eq:cl:all} in
case \ref{assump:general:a1} and \ref{assump:general:a2a} or
\ref{assump:general:a2b} are satisfied. 
\subsection{Estimates under Assumptions \ref{assump:general:a1} and \ref{assump:general:a2} }\label{subsec:A2ab}
\begin{lemma}\label{lemma:xest:case:1}
  Let \ref{assump:general:a1} and \ref{assump:general:a2a} or \ref{assump:general:a2b} be satisfied. Then there exists a constant $\clsym_1>0$ independent of $t$ and $\bs{k}$ such
  that for all $k\ge n+1$ the inequalities
  \begin{align}\label{eq:xkest}
    \begin{split}
      |x_k(t)| &\le e^{\lambda_{k}t} |x_k(0)| + \clsym_1\frac{\|\bs{k}\|_2}{k^\alpha}
      \|(\bs{x}_s-\bs{e}_s)_{[0,t]}\|_{2,\infty}\\
      &\le e^{\lambda_{n+1}t} |x_k(0)| + \clsym_1\frac{\|\bs{k}\|_2}{k^\alpha}
      \|(\bs{x}_s-\bs{e}_s)_{[0,t]}\|_{2,\infty}
      \end{split}
  \end{align}
  hold, where $\alpha=1$ in case \ref{assump:general:a2b} holds. In particular, this implies the existence of $\clsym_2>0$ with
  \begin{multline}\label{eq:xfest:l2:case:1}
    \|\bs x_f(t)\|_2 \le e^{\lambda_{n+1}t}\|\bs x_f(0)\|_2\\ + \frac{\clsym_2}{\sqrt{n^\alpha}} \|\bs{k}\|_2\|(\bs{x}_s-\bs{e}_s)_{[0,t]}\|_{2,\infty}
  \end{multline}
and in case that \ref{assump:general:a2a} holds, additionally
  \begin{multline}\label{eq:xfest:l1:case:1}
    \|\bs x_f(t)\|_1 \le e^{\lambda_{n+1}t}\|\bs x_f(0)\|_1\\ + \clsym_1\zeta(\alpha,n+1)\|\bs{k}\|_2 \|(\bs{x}_s-\bs{e}_s)_{[0,t]}\|_{2,\infty}.
  \end{multline}
  Herein, $\zeta(\alpha,n+1)=\sum_{k=n+1}^{\infty} 1/k^{\alpha}=\sum_{k=1}^{\infty} 1/(k+n)^{\alpha}$ denotes the Hurwitz zeta function with $\alpha>1$ from \ref{assump:general:a2a}. 
\end{lemma}
\begin{proof}
  By the variation of constants formula we obtain the estimate
  \begin{align*}
    |x_k(t)| & \le e^{\lambda_kt} |x_k(0)|\\
             &+ \int_0^t e^{\lambda_k(t-\tau)} |b_k| \|\bs{k}\|_2 \|\bs{x}_s(\tau)-\bs{e}_s(\tau)\|_{2}d\tau \\
             & \le  e^{\lambda_kt} |x_k(0)| + \bigg\vert\frac{b_k}{\lambda_k}\bigg\vert \|\bs{k}\|_2 \|\bs{x}_s-\bs{e}_s\|_{2,\infty}.
  \end{align*}
The inequalities in \eqref{eq:xkest} then follow directly from
\ref{assump:general:a2a} with $\clsym_1=d_1$ or
\ref{assump:general:a2b} with {$\clsym_1=\cssym_1 / \cssym_3$} and the fact that
$\lambda_k\le\lambda_{n+1}$ for $k\ge n+1$ due to
\ref{assump:general:a1}.

The additional inequalities \eqref{eq:xfest:l2:case:1} and \eqref{eq:xfest:l1:case:1} follow by
  taking the $\ell_2$- or $\ell_1$-norm, respectively, 
  of the
  expressions on both sides, making use of the triangle
  inequality. For \eqref{eq:xfest:l2:case:1} we additionally use the inequality
  \[ \sum_{k=n+1}^\infty \frac1{k^{2\alpha}} \le \frac1{n^{2(\alpha-1)}} \sum_{k=n+1}^\infty \frac1{k^2} \le \frac1{n^{2(\alpha-1)}} \frac{\pi^2}{6} {\frac{1}{n} = \frac{\pi^2}{6} \frac{1}{n^{2\alpha-1}}}
  \]
  as well as $n^{2\alpha-1} \ge n^\alpha$ and for \eqref{eq:xfest:l1:case:1} we use the
  definition of the Hurwitz zeta function with
  $\alpha>1$ and $n\geq 1$ if \ref{assump:general:a2a} holds. 
\end{proof}

\subsection{Estimates under Assumptions \ref{assump:general:a1} and \ref{assump:general:a2a}}\label{subsec:A2a}

We now define the quantities
\begin{align} z_k := \cs{k} x_k, \quad k\geq n+1,\label{eq:zVa}\end{align} 
so that
\begin{align}
\label{eq:Cf_times_x:z1}
  y_f = {C\bs{x}_f} = \cv_f\bs{x}_f = \sum_{k=n+1}^{\infty} z_k 
\end{align}
holds provided the infinite sum is absolutely convergent, i.e.,
$\|\bs{z}\|_1<\infty$. The following input-to-output stability (IOS) estimate for $\|\bs{z}\|_1$ ensures this absolute convergence.
\begin{lemma}
  \label{lem:ios_estimate:z:case:1}
  Let \ref{assump:general:a1} and \ref{assump:general:a2a} hold. Then there are constants $\clsym_a^1,\clsym_a^2>0$ such that with
    $\eta_n(t)=\sum_{k=n+1}^{\infty}e^{\lambda_k t}$ for all $t > 0$ we obtain 
    \begin{align*}
      \|\bs{z}(t)\|_1 \leq \clsym_a^1\eta_n(t)\|\bs{x}_f(0)\|_2 +
      \frac{\clsym_a^2}{n^{{1-\alpha}}} \|(\bs{x}_s-\bs{e}_s)_{[0,t]}\|_{2,\infty}. 
    \end{align*}
\end{lemma}
\begin{proof}
  Taking into account the definition of the Hurwitz zeta function note that
  \begin{align*}
    \clsym_1\zeta(\alpha,n+1)\|\bs{k}\|_2 =
    \frac{1}{n^{{\alpha-1}}}\clsym_1\|\bs{k}\|_2{\frac1n}\sum_{k=1}^{\infty}
    \bigg(\frac{n}{n+k}\bigg)^\alpha. 
  \end{align*}
  The second summand on the right hand side of the inequality then results by summing up the second terms in \eqref{eq:xfest:l1:case:1}. Since
  \begin{align*} \sum_{k=1}^{\infty}
    \bigg(\frac{n}{n+k}\bigg)^\alpha & = \sum_{k=1}^{\infty}
    \bigg(\frac{1}{1+k/n}\bigg)^\alpha \\ & \le \sum_{k=1}^{\infty}
    \bigg(\frac{1}{1+\lfloor k/n \rfloor}\bigg)^\alpha \le n \underbrace{\sum_{k=0}^{\infty}
    \bigg(\frac{1}{1+k}\bigg)^\alpha}_{< \infty \text{ since } \alpha>1},\end{align*}
     defining
  $\clsym_a^2:=\cssym_2\clsym_1\|\bs{k}\|_2\frac1n \sum_{k=1}^{\infty} (n/(n+k))^\alpha <
  \infty$ with $\cssym_2$ from \ref{assump:general:a2a} yields a constant that is 
 independent of $n$.

  To derive the first summand, consider \eqref{lemma:xest:case:1} and note that
  \begin{align*}
    \sum_{k=n+1}^{\infty}e^{\lambda_k t} \vert x_k(0)\vert \leq
    \sum_{k=n+1}^{\infty}e^{\lambda_k t} \|\bs{x}_f(0)\|_2
  \end{align*}
  as $\vert x_k(0)\vert\leq \|\bs{x}_f(0)\|_2$. The sum
  $\sum_{k=n+1}^\infty e^{\lambda_{k}t}$ with
  $\ldots<\lambda_{k+1}<\lambda_k<0$ for $k\geq n+1$ converges
  absolutely to a function $\eta_n(t)$ fulfilling
  $\lim_{t\to 0}\eta_n(t) = \infty$ and
  $\lim_{t\to \infty}\eta_n(t) = 0$ with exponential
  convergence. Let $a_k(t)=e^{\lambda_{k}t}$ for $k\geq n+1$, then
  absolute convergence for $t>0$ becomes apparent as
  $e^{\lambda_kt}\le q/|\lambda_k|$ for suitable $q>0$
  and $1/|\lambda _k|$ is absolutely convergent by \ref{assump:general:a1}.
  For $t=0$ we have
  $e^{\lambda_kt}=1$ for all $k$ 
  so that the series approaches infinity. 
  This yields the first summand on the right hand side of the estimate with $\clsym_a^1 = \cssym_2$ from \ref{assump:general:a2a}.
\end{proof}
\subsection{Estimates under Assumptions \ref{assump:general:a1}
    and \ref{assump:general:a2b}}\label{subsec:A2b}
In this section we derive a counterpart for Lemma \ref{lem:ios_estimate:z:case:1} in case \ref{assump:general:a1} and \ref{assump:general:a2b} are satisfied. The difficulty here is that Lemma \ref{lemma:xest:case:1} does not give us an immediate estimate for the $\ell_1$-norm $\|\bs x_f(t)\|_1$ if \ref{assump:general:a2a} does not hold. In order to circumvent this problem we have to use another definition of $z_j$.
We start with an auxiliary lemma. 

\begin{lemma} Let \ref{assump:general:a1} and \ref{assump:general:a2b} hold. Then there is a constant $\clsym_3>0$ such that
\[ \int_0^t \bigg| \sum_{k\in S_j} \cs{k} b_k e^{\lambda_k \tau} \bigg| d\tau \le \frac{\clsym_3}{j^2}  \]
for all $t\ge 0$ and all $j\in\N$ with $j\ge n/\cssym_4$ with $n$ and $\cssym_4$ from \ref{assump:general:a2b}.
\label{lemma:aux}\end{lemma}
\begin{proof} Abbreviate $a_k =\cs{k} b_k$ and consider the function 
\[ \tau \mapsto h_j(\tau) := \sum_{k\in S_j} a_k e^{{\lambda_k} \tau}. \]
This function is continuous and, since $S_j$ has $s$ elements, according to Descartes' rule of signs \cite[Theorem 3.1]{Jame06} it has $\sigma\le s$ zeros $t_1, \ldots, t_{\sigma}>0$, which we number in ascending order. This means that 
\[ t\mapsto H_j(t) := \int_{0}^{t} h_j(\tau) d\tau \] 
has at most $\sigma\le s$ local maxima and minima $H_j(t_l)$, $l=1,\ldots,\sigma$ and $h_j$ can change its sign only at the times $t_l$. 
Since $h_j(t)\to 0$ as $t\to\infty$ exponentially fast, the limit $H_j^\infty:=\lim_{t\to\infty} H_j(t)$ exists and is finite. 
Let 
\[ H_j^+:=\max\{0,H_j(t_1),\ldots,H_j(t_\sigma),H_j^\infty\} \] 
and 
\[ H_j^-:=\min\{0,H_j(t_1),\ldots,H_j(t_\sigma),H_j^\infty\}. \]
Then
\[ \left|\int_{s_1}^{s_2} h_j(\tau) d\tau\right| = | H(s_2)-H(s_1) |\le H_j^+ - H_j^-\]
for all $0\le s_1\le s_2$. 

Now we pick an arbitrary $t>0$, let $\sigma_t\le \sigma$ be the largest index with $t_{\sigma_t}<t$ and set $\tau_0:=0$, $\tau_l:=t_l$ for $l=1,\ldots,\sigma_t$ and $\tau_{\sigma_t+1}:= t$. Then, since $h$ can only change sign at the times $t_l$ we get 
\begin{eqnarray*} 
 \int_0^t | h_j(\tau) | d\tau & = & \sum_{l=0}^{\sigma_t} \int_{\tau_l}^{\tau_{l+1}} | h_j(\tau) | d\tau \\
& = & \sum_{l=0}^{\sigma_t} \left| \int_{\tau_l}^{\tau_{l+1}}  h_j(\tau)  d\tau\right| \\
& \le & \sum_{l=0}^{\sigma_t} H_j^+ - H_j^- \; \le \; s(H_j^+ - H_j^-).
\end{eqnarray*}
It thus suffices to prove that there is $\cssym_6>0$ with 
$|H_j(t_l)| \le \cssym_6 j^{-2}$ for all $l=1,\ldots,\sigma$ and $\lim_{t\to\infty}|H_j(t)|\le \cssym_6 j^{-2}$. Then the assertion follows with $\clsym_3=2s\cssym_6$.

To this end, observe that 
\[ H_j(t) = \sum_{k \in S_j} \frac{a_k}{p_k} e^{p_k t} - \frac{a_k}{p_k}.\] 
Hence, 
\[ \lim_{t\to\infty} |H_j(t)| = \bigg|\sum_{k\in S_j} \frac{a_k}{p_k} \bigg| \le \frac{\cssym_5}{j^2}.\]
For the $t_l$ we obtain 
\begin{eqnarray*}
H_j(t_l) & = & \sum_{k \in S_j} \frac{a_k}{p_k} e^{p_k t_l} - \frac{a_k}{p_k}\\
& \le & \bigg|\sum_{k \in S_j} \frac{a_k}{p_k} e^{p_k t_l} \bigg| + \bigg|\sum_{k\in S_j} \frac{a_k}{p_k} \bigg|.
\end{eqnarray*}
The second term is bounded by $\frac{\cssym_5}{j^2}$, as above. For the first term, observe that $h_j(t_l)=0$ implies 
\[ a_{\hat k} e^{p_{\hat k} \tau} = - \sum_{k\in S_j\setminus\{\hat k\}} a_k e^{p_k \tau},\]
where $\hat k$ is an arbitrary element in $S_j$. This yields
\begin{eqnarray*}
\bigg|\sum_{k \in S_j} \frac{a_k}{p_k} e^{p_k t_l} \bigg| 
& = & \bigg|\sum_{k\in S_j\setminus\{\hat k\}}  \left(\frac{a_k}{p_k} - \frac{a_k}{p_{\hat k}}\right) e^{p_k t_l}\bigg|\\
& \le & \sum_{k\in S_j\setminus\{\hat k\}}  a_k \left|\frac{1}{p_k} - \frac{1}{p_{\hat k}}\right|\\
& \le & \sum_{k\in S_j\setminus\{\hat k\}}  \cssym_1k \gamma(k_1) k^{-3}\\
& \le & \sum_{k\in S_j\setminus\{\hat k\}}  \cssym_1\gamma(k_1) k^{-2}\\
& \le & \frac{(s-1)\cssym_1\gamma(k_1)}{\cssym_4^2}\frac1{j^2}.
\end{eqnarray*}
All in all, we obtain the desired bound on $|H_j(t_l)|$ with $\cssym_6 = \frac{(s-1)\cssym_1\gamma(k_1)}{\cssym_4^2} + \cssym_5$.
\end{proof}

Using the notation of Lemma \ref{lemma:aux} we consider the quantities 
\begin{align}
  z_{j} := \sum_{k\in S_j} \cs{k} x_k \label{eq:zVb}
\end{align}
for $j\in\N$. We note that then for $n=\min S_{j_0}$ we obtain
\begin{align}
  \label{eq:Cf_times_x:z}
  y_f = {C\bs{x}_f} =  \cv_f\bs{x}_f = \sum_{j=j_0}^\infty z_{j}
\end{align}
holds, provided the infinite sum is absolutely convergent, i.e., provided $\|\bs{z}\|_1<\infty$ for $\bs{z} = (z_{j_0},z_{j_0+1},\ldots)^T$.

\begin{lemma}
  Let \ref{assump:general:a1} and \ref{assump:general:a2a} hold. Then there is a constant $\clsym_4>0$ independent of $n$, such that
  for all $j_0\in\N$ with $n=j_0\cssym_4-1$ satisfying $\lambda_{n+1}<0$, and all $t\ge 0$ the inequality
  \begin{eqnarray*}
    |z_{j}(t)| & \le & \clsym_4\bigg|\sum_{k\in S_j}e^{\lambda_{k}t} x_{k}(0)\bigg| \\
    && \qquad + \frac{\clsym_4}{j^2} \|\bs{k}\|_2 \|(\bs{x}_s-\bs{e}_s)|_{[0,t]}\|_{2,\infty}
  \end{eqnarray*}
  holds. 
  \label{lemma:zest:new}\end{lemma}
\begin{proof} 
By the variation of constants formula we obtain the
  estimate
  \begin{align*}
    &|z_{j}(t)| =  \bigg|\sum_{k\in S_j}\cs{k} x_{k}(t)\bigg| \\
    &\le  \bigg|\sum_{k\in S_j}\cs{k} e^{\lambda_{k}t} x_{k}(0)\bigg|\\
    & \quad+ \; \bigg| \int_0^t \sum_{k\in S_j}\cs{k} e^{\lambda_{k}(t-\tau)} b_{k} {\cKT} (\bs{x}_s(\tau)-\bs{e}_s(\tau))\bigg|d\tau \\
    & \le   \bigg|\sum_{k\in S_j}\cs{k} e^{\lambda_{k}t} x_{k}(0)\bigg|\\
    & \quad+ \; \int_0^t \bigg| \sum_{k\in S_j}\cs{k} b_k e^{\lambda_{k}(t-\tau)} \bigg|d\tau  \, \|\bs{k}\|_2 \|(\bs{x}_s-\bs{e}_s))|_{[0,t]}\|_{2,\infty}.
  \end{align*}
  Observing that
  \[ \int_0^t \bigg| \sum_{k\in S_j}\cs{k} b_k
    e^{\lambda_{k}(t-\tau)} \bigg|d\tau = \int_0^t \bigg| \sum_{k\in
      S_j}\cs{k} b_k e^{\lambda_{k}(\tau)} \bigg|d\tau,  \]
  the assertion follows from Lemma \ref{lemma:aux} with $\clsym_4 = \cssym_2$.
\end{proof}

Based on Lemma \ref{lemma:zest:new} we can obtain the counterpart of Lemma \ref{lem:ios_estimate:z:case:1}.

  \begin{lemma}\label{lemma:znormest}
    Let \ref{assump:general:a1} and \ref{assump:general:a2a} hold with $n=j_0\cssym_4-1$, $j_0\in\N$. Then there are constants $\clsym_b^1,\clsym_b^2>0$ such that for 
      $\eta_n(t)=\sum_{j=n+1}^\infty e^{\lambda_{j}t}$ the inequality
      \begin{align*}
        \|\bs{z}(t)\|_1 \le \clsym_b^1\eta_n(t)\|\bs{x}_f(0)\|_2 + \frac{\clsym_b^2}{n}
        \|(\bs{x}_s-\bs{e}_s)|_{[0,t]}\|_{2,\infty}
      \end{align*}
      holds for all $t > 0$.  
  \end{lemma}
\begin{proof}
The inequality follows by summing the right hand sides of Lemma \ref{lemma:zest:new} over
    $j \ge j_0$. Then the first terms sum up to 
    \begin{align*} \sum_{j=j_0}^\infty \clsym_4\bigg|\sum_{k\in S_j}e^{\lambda_{k}t} x_{k}(0)\bigg| & \le \clsym_4 \underbrace{\sum_{k=n+1}^\infty |e^{\lambda_{k}t}|}_{=\eta_n(t)} \underbrace{|x_{k}(0)|}_{\le \|\bs{x}_f(0)\|_2}\\
    & \le \clsym_b^1\eta_n(t) \|\bs{x}_f(0)\|_2\end{align*}
    with $\clsym_b^1 = \clsym_4$. 
    Using the same inequality as in the proof of Lemma \ref{lemma:xest:case:1}, we can estimate 
    \[ \sum_{j=j_0}^\infty \frac 1{j^2} \le \frac{\pi^2}{6j_0} = \frac{\cssym_4\pi^2}{6(n+1)} \le \frac{\cssym_4\pi^2}{6n},\]%
    and thus the second terms sum up to $\frac{\clsym_b^2}{n}
        \|(\bs{x}_s-\bs{e}_s)|_{[0,t]}\|_{2,\infty}$ if we set 
  \[ \clsym_b^2 := \frac{\cssym_4\pi^2}{6}\|\bs{k}\|_2. \] 
    This shows the claim.
\end{proof}
\subsection{A further estimate under Assumption \ref{assump:general:a1}}\label{subsec:A1}
In order to apply a small-gain argument, we also need an estimate for the slow subsystem. The following lemma yields this result if \ref{assump:general:a1} holds.
\begin{lemma}
  Let \ref{assump:general:a1} hold and let $\tilde{\bs{x}}_s=(\bs{x}_s,\bs{e}_s)$. There exists $\clsym_s^1,\clsym_s^2>0$
  such that the inequality 
  \[
    \|\tilde{\bs{x}}_s(t)\| \le \clsym_s^1 e^{- \tilde\lambda t}\|\tilde{\bs{x}}_s(0)\| + \clsym_s^2\int_0^{t} \|\bs{z}(\tau)\|_1 d\tau 
  \]
  holds for all $t\ge 0$ with $\tilde\lambda>0$ defined after \eqref{eq:c1tilde}.
  \label{lemma:iss_slow}
\end{lemma}
\begin{proof}
  The variation of constants formula applied to \eqref{eq:c1tilde}
  in view of \eqref{eq:Cf_times_x:z1} or \eqref{eq:Cf_times_x:z} and the definition of $\bs{z}$ provide
  \begin{align*}
    \|\tilde{\bs{x}}_s(t)\| &= \left\|e^{\widetilde At} \tilde{\bs{x}}_s(0) + \int_0^t
                       e^{\widetilde A(t-\tau)} \tilde{\oL} \cv_f\bs{x}_f(\tau) d\tau\right\|\\
                     &  \leq \left\| e^{\widetilde At} \tilde{\bs{x}}_s(0)\right\| + \left\|\int_0^t e^{\widetilde A(t-\tau)} \tilde{\oL} \|\bs z(\tau)\|_1 d\tau\right\|.
  \end{align*}%
Note that the norm of $\tilde{\oL}$ satisfies the same $n$-independent bound as the norm of $\oL$.
Setting 
  \[ \clsym_s^2 := \int_0^\infty \|e^{\widetilde A(t-\tau)}\| \|
      \tilde{\oL}\| d\tau < \infty, \]
  the integral term can be estimated via
  \begin{align*}
    &\left\|\int_0^t e^{\widetilde A(t-\tau)} \tilde{\oL} \|\bs
      z(\tau)\|_1 d\tau \right\|\\
    & \le  \int_0^{t} \|e^{\widetilde A(t-\tau)}\| \|
      \tilde{\oL}\| \|\bs z(\tau)\|_1 d\tau \; \le  \;
      \clsym_s^2\int_0^t \|\bs z(\tau)\|_1 d\tau.
  \end{align*}
  Together with the definition of $\tilde\lambda>0$ and $\widetilde M>0$ after \eqref{eq:c1tilde} this yields the claim with $\clsym_s^1 = \widetilde M$.
\end{proof}
\section{Main result}\label{sec:main}
The proof of our following main stability theorem is inspired by
\cite{BLJZ18,JiTP94}, which give stability criteria for
interconnected systems based on input-to-output stability (IOS) and
unbounded observability (UO) properties. There are, however, two
important differences. On the one hand, the proof is significantly
simplified here as we can make use of the linearity of the system. On
the other hand, the term $\eta_n(t)$ on the right hand sides of the estimates in Lemma
\ref{lem:ios_estimate:z:case:1} and \ref{lemma:znormest} tends to $\infty$ as $t\to 0$, which requires
a slightly more involved treatment of this term.
\begin{theorem} \label{thm:main}
  Consider the closed-loop system
  \eqref{eq:cl:all} satisfying \ref{assump:general:a1}
  and \ref{assump:general:a2a} or \ref{assump:general:a2b}. Then, for
  all sufficiently large $n$ there are constants $C>0$ and
  $\lambda>0$ such that the slow subsystem satisfies the estimate
  \begin{align}
    \label{eq:thm:main}
    \|\tilde{\bs x}_s(t)\|_2 \le Ce^{-\lambda t} \|\tilde{\bs x}_s(0)\|_2 +  Ce^{-\lambda t}\|\bs x_f(0)\|_2
  \end{align}
  for all $t\ge 0$ and all initial conditions $\tilde{\bs x}_s(0)\in\mathbb{R}^n$ and $\bs x_f(0)\in \ell^2$. 
  In particular, this implies exponential stability of \eqref{eq:cl:all} in $\ell_2$.
\end{theorem}
\begin{proof}
  Lemma \ref{lemma:iss_slow} yields 
  \[
    \|\tilde{\bs{x}}_s(t)\|_2 \le  \clsym_s^1 e^{- \tilde\lambda
      t}\|\tilde{\bs{x}}_s(0)\|_2 + \clsym_s^2\int_0^{t}
    \|\bs{z}(\tau)\|_1 d\tau. 
  \]
  The $\bs z$-dependent term can be suitably bounded. Using
    Lemma \ref{lem:ios_estimate:z:case:1} in case \ref{assump:general:a2a} holds or Lemma \ref{lemma:znormest} in case \ref{assump:general:a2b} holds we obtain
    \begin{align*}
    &\int_0^{t} \|\bs z(\tau)\|_1 d\tau\\
    &\qquad \leq
      \clsym_*^1\int_0^{t} \eta_n(\tau) d\tau \|\bs{x}_f(0)\|_2 + \frac{\clsym_*^2}{n^{{\beta}}}
    \|(\bs{x}_s-\bs{e}_s)|_{[0,t]}\|_{2,\infty}
  \end{align*}
  with $*=a$ and {$\beta=\alpha-1>0$} if \ref{assump:general:a2a} holds and $*=b$ and {$\beta=1$} if \ref{assump:general:a2b} holds.
  The function $\eta_n(t)=\sum_{j=n+1}^\infty e^{\lambda_{j}t}$ is
  obtained as the limit of the sequence 
  $(\mathfrak{a}_p(t))_{p\geq n+1}$ 
  with $\mathfrak{a}_p(t)=\sum_{j=n+1}^{p}
  e^{\lambda_{j}t}$. Observing that $\mathfrak{a}_p(t)$ is monotone,
  positive and absolutely integrable, Beppo Levi's monotone
  convergence theorem implies that  
  \begin{align*}
    \int_0^{t} \eta_n(\tau) d\tau
    &= \lim_{p\to\infty}\int_0^{t}
    \mathfrak{a}_p(\tau) d\tau  = \lim_{p\to\infty}
      \int_{0}^{t}\sum_{j=n+1}^{p}
      e^{\lambda_j\tau} d\tau\\
    &=
    \sum_{j=n+1}^{\infty}
      \frac{1-e^{\lambda_j{t}}}{\vert\lambda_j\vert}\leq
    \sum_{j=n+1}^{\infty}
    \frac{1}{\vert \lambda_j\vert}\leq M_\lambda,
  \end{align*}
  where the last inequality follows from Assumption
  \ref{assump:general:a1}. 
  As a consequence
  \begin{align*}
    &\int_0^{t} \|\bs z(\tau)\|_1 d\tau\leq
    \clsym_*^1M_\lambda\|\bs{x}_f(0)\|_2 + \frac{\clsym_*^2}{n^{{\beta}}}
      \|(\bs{x}_s-\bs{e}_s)|_{[0,t]}\|_{2,\infty}.
  \end{align*}
Using $\|(\bs x_s - \bs e_s)|_{[0,t]}\|_{2,\infty} \le 
\sqrt 2 \|\tilde{\bs x}_s|_{[0,t]}\|_{2,\infty}$ 
we obtain
\[ \|\tilde{\bs{x}}_s(t)\|_2 \le  \clsym_s^1 e^{- \tilde\lambda
    t}\|\tilde{\bs{x}}_s(0)\|_2 + \clsym_5\| \bs x_f(0)\|_2  +
  \frac{\clsym_6}{n^{{\beta}}}\|\tilde{\bs x}_s|_{[0,t]}\|_{2,\infty}\]
with $\clsym_5=\clsym_s^2M_\lambda \clsym_*^1$ and $\clsym_6=\clsym_s^2\sqrt 2 \clsym_*^2$.
Taking the supremum for $t\in [0,\tau]$ of both sides of this inequality yields
\[ \|\tilde{\bs x}_s|_{[0,\tau]}\|_{2,\infty} \le  \clsym_s^1 \|\tilde{\bs{x}}_s(0)\|_2 + \clsym_5\| \bs x_f(0)\|_2  +  \frac{\clsym_6}{n^{{\beta}}}\|\tilde{\bs x}_s|_{[0,\tau]}\|_{2,\infty}.\]
Choosing $n$ so large that $\clsym_6/n^{{\beta}}\le 1/2$ holds and subtracting the last term on the right hand side implies 
\begin{equation} \|\tilde{\bs x}_s|_{[0,\tau]}\|_{2,\infty} \le  2\clsym_s^1 \|\tilde{\bs{x}}_s(0)\|_2 + 2\clsym_5\| \bs x_f(0)\|_2. 
\label{eq:between}
\end{equation}
Moreover, from \eqref{eq:xfest:l2:case:1} we obtain
\[ \|\bs x_f(t)\|_2 \le e^{\lambda_{n+1}t}\|\bs x_f(0)\|_2 +\frac{\clsym_2}{\sqrt {n^\alpha}}{\|\bs{k}\|_2}  \|\tilde{\bs{x}}_s|_{[0,t]}\|_{2,\infty}.\]
We increase $n$ further, if necessary, such that $2\clsym_s^1\clsym_5/n^{{\beta}} \le
1/(8(\clsym_5+1))$, $2\clsym_5^2/n^{{\beta}} \le 1/(8(\clsym_5+1))$, $2\clsym_s^1\clsym_5/\sqrt{n^{{\beta}}} \le
1/(8(\clsym_5+1))$, and $2\clsym_5\clsym_2{\|\bs{k}\|_2} /\sqrt{n^\alpha} \le 1/(8(\clsym_5+1))$ hold. This
implies
\begin{multline*}
  \|\tilde{\bs{x}}_s(t)\|_2 \le  \clsym_s^1 e^{- \tilde\lambda
    t}\|\tilde{\bs{x}}_s(0)\|_2 + \clsym_5\| \bs x_f(0)\|_2\\
  +  \frac{1}{8(\clsym_5+1)}(\|\tilde{\bs{x}}_s(0)\|_2+\| \bs x_f(0)\|_2)
\end{multline*}
and
\begin{multline*}
  \|\bs x_f(t)\|_2 \le e^{\lambda_{n+1}t}\|\bs x_f(0)\|_2 \\+ \frac1{8(\clsym_5+1)}(\|\tilde{\bs{x}}_s(0)\|_2+\| \bs x_f(0)\|_2) .
\end{multline*}
Choosing $\tau>0$ such that $\clsym_s^1 e^{- \tilde\lambda \tau}\le
1/(8(\clsym_5+1))$ and $e^{\lambda_{n+1}\tau}\le 1/(8(\clsym_5+1))$,
we obtain
\begin{align*}
  \|\tilde{\bs{x}}_s(\tau)\|_2 &\le  \frac1{4(\clsym_5+1)}
                                 \|\tilde{\bs{x}}_s(0)\|_2 +
                                 \left(\clsym_5+1\right)\| \bs
                                 x_f(0)\|_2\\
  \|\bs x_f(\tau)\|_2 & \le \frac{1}{4(\clsym_5+1)}\|\bs x_f(0)\|_2 +
                        \frac{1}{4(\clsym_5+1)}\|\tilde{\bs{x}}_s(0)\|.
\end{align*}
Since the system is time-invariant, for $t_k=k\tau$, $k\in\N_0$, by
the same reasoning we obtain the inequalities
\begin{align*}
  \|\tilde{\bs{x}}_s(t_{k+1})\|_2 &\le  \frac1{4(\clsym_5+1)}
                                    \|\tilde{\bs{x}}_s(t_k)\|_2 +
                                    \left(\clsym_5+1\right)\| \bs
                                    x_f(t_k)\|_2 \\
  \|\bs x_f(t_{k+1})\|_2 &\le \frac{1}{4(\clsym_5+1)}\|\bs
                           x_f(t_k)\|_2 +
                           \frac{1}{4(\clsym_5+1)}\|\tilde{\bs{x}}_s(t_k)\|. 
\end{align*}
Considering equality the two equations can be considered as coupled
difference equations that admit a closed-form solution. Based on this
the following estimates are obtained for all $k\ge 1$
\begin{align*}
  \|\tilde{\bs{x}}_s(t_{k})\|_2 &\le
                                 2\bigg(\frac{3}{4}\bigg)^k(\|\tilde{\bs{x}}_s({t_0})\|_2+(\clsym_5+1)\|\bs
                                  x_f({t_0})\|_2)\\
  \|\bs x_f(t_{k})\|_2 &\le 2\bigg(\frac{3}{4}\bigg)^k(\|\bs x_f({t_0})\|_2+\|\bs x_s({t_0})\|_2).
\end{align*}
Between the times $t_k$, with the same arguments as those leading to \eqref{eq:between} we obtain
\[ \|\tilde{\bs x}_s|_{[t_k,t_{k+1}]}\|_{2,\infty} \le  2\clsym_s^1 \|\tilde{\bs{x}}_s(t_k)\| + 2\clsym_5\| \bs x_f(t_k)\|_2. \]
Together this yields the claimed inequality \eqref{eq:thm:main} with
$\lambda = \log(4/3)$ and
$C=\max\{\clsym_s^1,\clsym_5\}(\clsym_5+1)16/3$. 

To prove exponential stability in $\ell_2$ of the overall
system, it remains to be shown that also $\|\bs x_f(t)\|_2$ has an
upper bound that is linear in the norms of the initial conditions
$\|\bs x_f(0)\|_2$ and $\|\tilde{\bs x}_s(t)\|_2$ and decays
exponentially. This follows by combining the already proved inequality
\eqref{eq:thm:main} with estimate \eqref{eq:xfest:l2:case:1} from
Lemma \ref{lemma:xest:case:1} as follows: We again pick $\tau>0$ and
consider the times $t_k = k \tau$. Applying \eqref{eq:xfest:l2:case:1}
on the interval $[t_k,t_{k+1}]$ yields 
  \begin{align*}
    & \|\bs x_f(t_{k+1})\|_2 \\ & \le e^{\lambda_{n+1}\tau}\|\bs x_f(t_k)\|_2 + \frac{\clsym_2}{\sqrt{n^{{\beta}}}} \|\bs{k}\|_2\|(\bs{x}_s-\bs{e}_s)_{[t_k,t_{k+1}]}\|_{2,\infty}\\
    & \le e^{\lambda_{n+1}\tau}\|\bs x_f(t_k)\|_2 + \frac{2\clsym_2}{\sqrt{n^{{\beta}}}} \|\bs{k}\|_2 \|(\tilde{\bs{x}}_s)_{[t_k,t_{k+1}]}\|_{2,\infty} \\
    & \le e^{\lambda_{n+1}\tau}\|\bs x_f(t_k)\|_2\\ & \quad + \frac{2\clsym_2}{\sqrt{n^{{\beta}}}} \|\bs{k}\|_2 (Ce^{-\lambda t_k} \|\tilde{\bs x}_s(0)\|_2 +  Ce^{-\lambda t_k}\|\bs x_f(0)\|_2)\\
    & = e^{\lambda_{n+1}\tau}\|\bs x_f(t_k)\|_2 + \widetilde C(e^{-\lambda t_k} \|\tilde{\bs x}_s(0)\|_2 +  e^{-\lambda t_k}\|\bs x_f(0)\|_2).
  \end{align*}
  with $\widetilde C = \frac{2\clsym_2}{\sqrt{n^{{\beta}}}} \|\bs{k}\|_2 C$. By induction we obtain
  \begin{align*} 
 &  \|\bs x_f(t_{k+1})\|_2  \le e^{\lambda_{n+1}t_{k+1}}\|\bs x_f(0)\|_2 \\
  & \quad + \widetilde C \sum_{j=0}^k e^{\lambda_{n+1}(t_{k}-t_j)} (e^{-\lambda t_j} \|\tilde{\bs x}_s(0)\|_2 +  e^{-\lambda t_j}\|\bs x_f(0)\|_2)\\
  & \le e^{\lambda_{n+1}t_{k+1}}\|\bs x_f(0)\|_2 \\
  & \quad + (k+1)\widetilde C e^{-\tilde\lambda t_{k}} (\|\tilde{\bs x}_s(0)\|_2 +  \|\bs x_f(0)\|_2)\\
  & \le e^{\lambda_{n+1}t_{k+1}}\|\bs x_f(0)\|_2 \\
  & \quad + \widehat C e^{-\hat\lambda t_{k+1}} (\|\tilde{\bs x}_s(0)\|_2 +  \|\bs x_f(0)\|_2).
  \end{align*}
  Here we used $\tilde\lambda = \min\{\lambda,-\lambda_{n+1}\}$ and the fact that for any $\hat\lambda<\tilde\lambda$ there is $M>0$ with $(k+1)e^{-\tilde\lambda t_k}\le Me^{-\hat\lambda t_{k+1}}$. 
  Between the times $t_k$ we can again use the same arguments as those leading to \eqref{eq:between}, establishing the desired exponential estimate
  \begin{equation} \|\bs x_f(t)\|_2 \le Ce^{-\lambda t}(\|\tilde{\bs x}_s(0)\|_2 +  \|\bs x_f(0)\|_2)\label{eq:xfexp}\end{equation}
  with suitably redefined $C,\lambda>0$ for all $t\ge 0$.
\end{proof}

\begin{remark} The particular value $\lambda = \log(4/3)$ in the proof could be changed to any arbitrary positive constant by adjusting the constants, the controller and the value of  $n$ appropriately. This means that the resulting exponential decay rate $-\tilde\lambda = - \min\{\lambda,-\lambda_{n+1}\}$ can be made as negative as desired. Note, however, that other performance measures such as the phase and gain margin of the closed-loop transfer function cannot be directly determined from our approach.
\end{remark}

\begin{corollary}
  Let the assumptions of Theorem \ref{thm:main} hold true. Then the
  closed-loop system composed of \eqref{eq:dr:abstract} with state
  feedback $u=-\cKT\hat{\bs{x}}_s$ evaluated using the observer
  \eqref{eq:obs} is exponentially stable in $X=(L^2(0,1))^N$.   
\end{corollary}
\begin{proof}
  The result follows by combining the Riesz basis property of the
  operator $A$ and the estimates \eqref{eq:thm:main} and \eqref{eq:xfexp}, which imply the existence of constants  $C>0$ and $\lambda>0$ so that 
  \begin{align*}
    \|\bs{x}(t)\|_X &= \bigg\|\sum_{k=1}^{\infty} x_k(t)\bs{\phi_k}\bigg\|_X=
                      \bigg(\sum_{k=1}^{\infty} (x_k(t))^2\bigg)^{\frac{1}{2}} \\
                    & \leq \|\bs{x}_s(t)\|_2 +
                      \|\bs{x}_f(t)\|_2\\
                    & \le  2Ce^{-\lambda t}(\|\tilde{\bs x}_s(0)\|_2 +  \|\bs x_f(0)\|_2)  \\
                    & \le 2Ce^{-\lambda t}(\|\bs x_s(0)\|_2 + \|\bs e_s(0)\|_2+  \|\bs x_f(0)\|_2)\\
                    & = 2Ce^{-\lambda t}(2\|\bs x(0)\|_X + \|\bs e_s(0)\|_2).
  \end{align*}
  Here in the last step we used the inequalities $\|\bs x_s(0)\|_2 \le
  \|\bs x(0)\|_X$ and $\|\bs x_f(0)\|_2 \le \|\bs x(0)\|_X$. This
  shows the claim.
\end{proof}
\begin{remark}\label{rem:mimo}
  In the MIMO-case, the operators $\bs{b}$ and $\cv$ are replaced by
  operators $\bs{B}$ with $m$ columns and $\bs{C}$ with $l$ rows,
  respectively. This implies that the vectors $\cK$ and $\oL$ become
  matrices $\bs{K}$ and $\bs{L}$ of appropriate dimensions. All
  proofs can be straightforwardly extended to this case if the
  following modifications are made in Assumption
  \ref{assump:general}, the respective lemmas and their proofs. 
  \begin{itemize}
  \item The vector norms $\|\cK\|$ and $\|\oL\|$ are replaced by the corresponding induced matrix norms $\|\bs{K}\|$ and $\|\bs{L}\|$.
  \item The fraction $| {b_k}/{\lambda_k}|$ is replaced
    by ${\|\bs{b}_k\|}/{|\lambda_k|}$, where $\bs{b}_k^T$ is now the $k$-th row of $\bs {B}$.
  \item The modulus $|\cs{k}|$ is replaced by the norm $\|\bs{c}_{k}\|$, where $\bs{c}_{k}$ is the $k$-th column of $\bs{C}$.
  \item The modulus $\vert\sum_{k\in S_j} {\cs{k} b_k}/{\lambda_k}\vert$ is replaced by the induced matrix norm $\|\sum_{k\in S_j} {\bs{c}_{k}\bs{b}^T_k}/{\lambda_k}\|$. Note that $\bs{c}_{k}\bs{b}^T_k$ is an $l\times m$-matrix in the MIMO-case. 
  \end{itemize}
\end{remark}
\section{Numerical computation of $n$}\label{sec:comp_n}
While our approach in principle allows for computing a bound on $n$
for which the inequalities required in the proof of Theorem
\ref{thm:main} hold, this bound will be very conservative. We can,
however, use a numerical approach that leads to a tighter bound: We
fix a second index $m>n$ and split the state of the fast subsystem
into
\[ \bs{x}_{f1}=[x_{n+1},\ldots, x_m]^T \mbox{ and } \bs{x}_{f2}=[x_{m+1},x_{m+2},\ldots]^T.\]
Then the overall closed loop system becomes
\begin{subequations}
  \label{eq:num:all}
  \begin{align}
  \dot{\bs{x}}_s &= \Lambda_s \bs{x}_s - \bs{b}_s \cKT(\bs{x}_s - \bs{e}_s)\label{eq:num1}\\
  \dot{\bs{e}}_s &= (\Lambda_s-\oL \cv_s) \bs{e}_s - \oL \cv_{f1}\bs{x}_{f1} - \oL \cv_{f2}\bs{x}_{f2} \label{eq:num2}\\
  \dot{\bs{x}}_{f1} &= \Lambda_{f1} \bs{x}_{f1} - \bs{b}_{f1} \cKT(\bs{x}_s - \bs{e}_s) \label{eq:num3}\\
  \dot{\bs{x}}_{f2} &= \Lambda_{f2} \bs{x}_{f2} - \bs{b}_{f2} \cKT(\bs{x}_s - \bs{e}_s) \label{eq:num4}.
  \end{align}
\end{subequations}
If we neglect the infinite-dimensional part $\bs x_{f2}$ of the fast dynamics, then we obtain
\begin{subequations}
  \label{eq:numa:all}
  \begin{align}
  \dot{\bs{x}}_s &= \Lambda_s \bs{x}_s - \bs{b}_s \cKT(\bs{x}_s - \bs{e}_s)\label{eq:numa1}\\
  \dot{\bs{e}}_s &= (\Lambda_s-\oL \cv_s) \bs{e}_s - \oL \cv_{f1}\bs{x}_{f1} \label{eq:numa2}\\
  \dot{\bs{x}}_{f1} &= \Lambda_{f1} \bs{x}_{f1} - \bs{b}_f \cKT(\bs{x}_s - \bs{e}_s) \label{eq:numa3}.
  \end{align}
\end{subequations}
Equation \eqref{eq:numa:all} defines a finite-dimensional LTI system, whose stability can be easily checked by
analyzing the eigenvalues of the overall system matrix. The question,
however, is, whether stability of \eqref{eq:numa:all}
implies stability of the true closed-loop system \eqref{eq:num:all}.

In order to see whether this is the case, we define
\[ \bs{x}_{sf} := \left(\begin{array}{l} \bs{x}_s \\
                       \bs{x}_{f1}\end{array}\right),  \;\;
   \bs{b}_{sf} := \left(\begin{array}{l} \bs{b}_s \\
                       \bs{b}_{f1}\end{array}\right),
\]
\[
   \tilde{\bs{k}} := \left(\begin{array}{l} \bs{k} \\
                       \bs{0}\end{array}\right), \;\;
                   \widetilde{\bs{c}}^T := \left(\bs{0}^T \; \cv_{f1}\right), \;\;
   \cv_{sf} := \left(\cv_s \; \cv_{f1}\right),
\]
and $\Lambda_{sf} := \text{diag}\{\Lambda_s,\Lambda_{f1}\}$. 
With these vectors and matrices, \eqref{eq:num:all} can be rewritten as
\begin{subequations}
  \label{eq:numx:all}
  \begin{align}
  \dot{\bs{x}}_{sf} &= \Lambda_{sf} \bs{x}_{sf} - \bs{b}_{sf} \ctKT(\bs{x}_{sf} - \bs{e}_{sf})\label{eq:num1a}\\
  \dot{\bs{e}}_{s} &= (\Lambda_{s}-\oL \cv_{s}) \bs{e}_{s} - \oL \widetilde{\bs{c}}^T\bs{x}_{sf} - \oL \cv_{f2}\bs{x}_{f2} \label{eq:num2a}\\
  \dot{\bs{x}}_{f2} &= \Lambda_{f22} \bs{x}_{f2} - \bs{b}_{f2} \ctKT(\bs{x}_{sf} - \bs{e}_{sf}) \label{eq:num4a}.
  \end{align}
\end{subequations}
while \eqref{eq:numa:all} becomes
\begin{subequations}
  \label{eq:numxa:all}
  \begin{align}
  \dot{\bs{x}}_{sf} &= \Lambda_{sf} \bs{x}_{sf} - \bs{b}_{sf} \ctKT(\bs{x}_{sf} - \bs{e}_{sf})\label{eq:numa1a}\\
  \dot{\bs{e}}_{s} &= (\Lambda_{s}-\oL \cv_{s}) \bs{e}_{s} - \oL \widetilde{\bs{c}}^T\bs{x}_{sf}\label{eq:numa2a}.
  \end{align}
\end{subequations}
Now one sees that \eqref{eq:numx:all} have a similar
structure as \eqref{eq:cl:all}.  Particularly, if \ref{assump:general:a1} and \ref{assump:general:a2a} or \ref{assump:general:a2b} hold, then the
subsystem \eqref{eq:num4a} satisfies all the requirements of the respective lemmata in Sections \ref{subsec:A2ab}--\ref{subsec:A2b} with $m$ in place of $n$. The
only assumption on subsystem \eqref{eq:num1a}--\eqref{eq:num2a} needed
for Theorem \ref{thm:main} is imposed in Lemma
\ref{lemma:iss_slow} in Section \ref{subsec:A1}. There it is assumed that its overall system
matrix, i.e., the matrix $\widetilde A$ in \eqref{eq:c1tilde} is
Hurwitz with $\max {\rm Re}(\lambda_i) < \bar\lambda<0$. Hence, if
this assumption is satisfied, then Theorem \ref{thm:main} can be
applied to \eqref{eq:num1a}--\eqref{eq:num4a} with $m$ in place of
$n$. This means that stability of \eqref{eq:numx:all}
(or, equivalently, of \eqref{eq:numxa:all}) for
sufficiently large $m$ and with
$\max {\rm Re}(\lambda_i) <\bar \lambda<0$ with $\bar\lambda$
independent\footnote{We note that $\bar\lambda$ must be independent of
  $m$ (resp.\ $n$) because the constant $\clsym_s^2$ from Lemma
  \ref{lemma:iss_slow}, which determines the size of the
  ``sufficiently large'' $n$ in the proof of Theorem \ref{thm:main} via the constant $\clsym_6$,
  depends on $\bar\lambda$.} of $m$ implies stability of
\eqref{eq:numx:all} and thus of
\eqref{eq:cl:all}.
This leads to the following numerical test in order to check whether a
certain number $n$ of modes taken into account in the controller is
sufficient for stabilization:
\begin{enumerate}
\item[(i)] Fix $n$ and compute $\rho_m=\max_i {\rm Re}(\lambda_i)$ for the eigenvalues $\lambda_i$ of the matrix
  governing the LTI system \eqref{eq:numa:all} for
  growing numbers of $m$  in order to find $\bar\lambda\in\R$ and $m_0\in\N$ such that $\rho_m\le \bar \lambda$ holds for all $m\ge m_0$.
\item[(ii)] If {$\bar \lambda<0$}, then system \eqref{eq:numx:all} is
exponentially stable for the given $n$.
\end{enumerate}
Clearly, by means of numerical computations it is not possible to {\em rigorously} ensure $\rho_m\le \bar \lambda$ for {\em all} $m\ge m_0$. However, often---as in the examples in the next section---convergence of $\rho_m$ for $m\to\infty$ can be observed numerically, which provides a strong evidence for the desired inequality since $\rho_m$ hardly changes anymore for large $m$.
\section{Simulation results}\label{sec:sim}
In the following the previous analysis and main results are evaluated
for three simulation scenarios covering both scalar and coupled
diffusion-reaction systems. 
\subsection{Scalar diffusion-reaction problem}
Based on the introductory problem formulation in Section
\ref{sec:prototype} boundary and in-domain control as well as sensing,
respectively, are considered to numerically evaluate the formulated
preliminaries and results. 
\subsubsection{Boundary control and {point} sensing}\label{subsec:sim:ex1}
\begin{figure*}[!h]
  \centering
  \subcaptionbox{{Values of $\rho_m$} for $n\in[3,8]$ when varying $m=n+1,\ldots,200$.\label{fig:ex:1:spectral_radius:mvar}}
  {\includegraphics[width=0.32\textwidth]{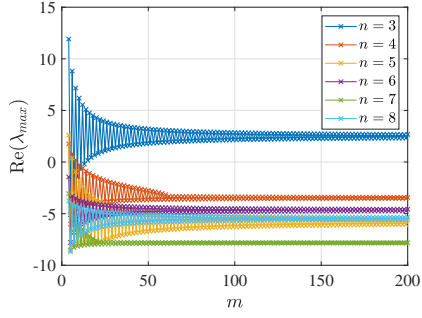}}
  \qquad
  \subcaptionbox{{Values of $\rho_m$} for $m=200$ and variation of the
    dimension $n$ of the slow subsystem.\label{fig:ex:1:spectral_radius:nvar}}
  {\includegraphics[width=0.32\textwidth]{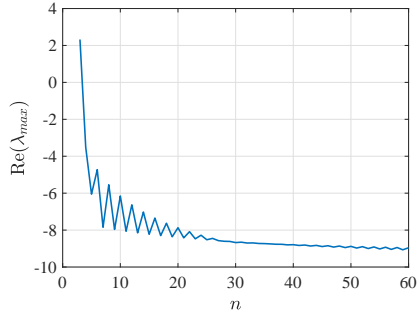}}
  \caption{Results for the scalar example of Section
    \ref{subsec:sim:ex1}.}
  \label{fig:ex:1:spectral_radius}
\end{figure*}
We illustrate the algorithm using the equation
\begin{subequations}
  \label{eq:ex1}
  \begin{align}
    &\partial_t x = \partial_z^2 x + r x,&&z\in(0,1),~t>0\\
    &\partial_z x\vert_{z=0}=0,\quad x\vert_{z=1} = u,&&t>0\\
    &x\vert_{t=0}=x_0,&&z\in[0,1]. 
  \end{align}
\end{subequations}
We set $r=15$, leading to the first three eigenvalues $\lambda_1
\approx 12.5326$, $\lambda_2 \approx  -7.2066$, and $\lambda_3
\approx -46.6850$. As output we use a point measurement at
$\xi=1/4$, leading to the components $c_k=\sqrt{2} \cos(\omega_k/4)$
of the operator $C$ and rendering the actuator/sensor configuration
non-collocated. The stabilizing feedback $\cK$ and the observer 
matrix $\oL$ are designed to shift the open-loop eigenvalues
$\lambda_k$ to the desired eigenvalues $\kappa_1 = -10$, $\kappa_2 =
-11$, and $\kappa_k=\lambda_k$ for $k\in [3,n]$ for the feedback
control and $\nu_1 = -15$, $\nu_2 = -16$, and $\nu_k=\lambda_k$ for
$k\in [3,n]$ for the observer error dynamics. Note that it is
subsequently not aimed at studying closed-loop performance but to
illustrate the main stability result. The verification of the
remaining conditions formulated in Assumption \ref{assump:general} can
be found in Remark \ref{rem:assump:prototype}.  

We have implemented and run the algorithm specified in Section
\ref{sec:comp_n} in MATLAB. Figure \ref{fig:ex:1:spectral_radius:mvar}
shows the resulting {$\rho_m$} when varying $n\in[3,8]$ for
$m=n+1,\ldots,200$. In all examples the numerical evidence strongly suggests that the values $\rho_m$ converge to fixed values and do not change their sign anymore after $m=200$. The results thus indicate that the controller
with observer based on the order $n=3$ is insufficient to stabilize
the system. For $n=4$ and $n=5$ stabilization is achieved for
sufficiently large $m$ and thus for the infinite-dimensional model. It
is furthermore interesting to observe a non-monotonic behavior in $n$
as the values of $\rho_m$ for $n=8$ are larger than the ones for $n=7$ and
even $n=5$ when $m\gg n+1$. This behavior is confirmed in Figure
\ref{fig:ex:1:spectral_radius:nvar}, where the $\rho_m$
computed for $m=200$ are shown depending on the order $n$. Obviously
stabilization is achieved for $n>4$ and a limit is approached that is
larger than the assigned smallest eigenvalue $\kappa_1 = -10$ with
decaying oscillatory behavior.
\subsubsection{In-domain control and sensing}\label{subsec:sim:ex2}
By replacing the boundary control in \eqref{eq:ex1} by an in-domain
control the problem reads
\begin{subequations}
  \label{eq:ex2}
  \begin{align}
    &\partial_t x = \partial_z^2 x + r x +bu,&&z\in(0,1),~t>0\\
    &\partial_z x\vert_{z=0}=0,\quad x\vert_{z=1} = 0,&&t>0\\
    &x\vert_{t=0}=x_0,&&z\in[0,1].      
  \end{align}
\end{subequations}
Given the Heaviside function $\sigma(\cdot)$ let
\begin{multline}
  \label{eq:ex2:rectangular_shape}
  f_{\zeta,\epsilon} =
  \frac{1}{2\epsilon}\big(\sigma(z-\zeta+\epsilon)-\sigma(z-\zeta-\epsilon)\big),\\0<\zeta-\epsilon<\zeta+\epsilon<1
\end{multline}
denote the rectangular shaped pulse centered at $z=\zeta$ with width
$2\epsilon$ and height $1/(2\epsilon)$. Furthermore let the spatial
input characteristics be given by $b=f_{\zeta,\epsilon}$ so that
$b(\cdot)$ approaches a Dirac delta function centered at $z=\zeta$
in the limit as $\epsilon\to 0$, which refers to a pointwise in-domain
control. The output is taken as a pointwise measurement at
position $\xi\in (0,1)$ so that $c_k=\sqrt{2}\cos(\omega_k\xi)$. 
Making use of the eigenvalue and eigenfunction
computations in Section \ref{sec:prototype}
we obtain $b_{1,k}=\langle b,\phi_k\rangle$ in terms of 
\begin{align*}
  b_{1,k} = \frac{2 \sqrt{2} \cos (\frac{\pi}{2} (2k-1)\zeta ) \sin (\frac{\pi}{2} (2 k-1) \epsilon)}{(2k-1) \pi \epsilon},
\end{align*}
which is used to confirm Assumptions \ref{assump:general:a1} and \ref{assump:general:a2a} instead of
$b_k$. Provided that the parameter pair $(\zeta, \epsilon,\xi)$ is
such that the slow subsystem~\eqref{eq:findim:slow} is stabilizable and
detectable, then \ref{assump:general:a0} and \ref{assump:general:a1}
are immediately fulfilled as in the example of Section
\ref{subsec:sim:ex1}. \ref{assump:general:a2a} follows since $\vert
b_{1,k}\vert \leq2 \sqrt{2}/((2k-1) \pi \epsilon)$ so that  
\begin{align*}
  \left\vert\frac{b_{1,k}}{\lambda_k}\right\vert\leq
  \frac{2 \sqrt{2}}{\left(\frac{(2 k-1)^2\pi^2}{4} -r\right) (2
  k-1)\pi \epsilon}\leq \frac{\frac{1}{\sqrt{2}\pi\epsilon}}{k^2}
\end{align*}
for $k\geq n$ with $n$ chosen sufficiently large depending on
$r$. With this, proceed as in Section \ref{subsec:sim:ex1} by determining the
stabilizing feedback $\cK$ and the observer matrix $\oL$ to shift
the open-loop eigenvalues $\lambda_k$ to the desired eigenvalues
$\kappa_1 = -10$, $\kappa_2 = -11$, $\kappa_k=\lambda_k$ for
$k\in[3,n]$ for the feedback control and $\nu_1 = -15$, $\nu_2 =
-16$,  $\nu_k=\lambda_k$ for $k\in[3,n]$ for the observer error
dynamics. By assigning $\zeta=0.7$, $\epsilon=0.05$ and $\xi=0.4$,
the resulting {$\rho_m$} for $n\in[3,8]$ and
$m=n+1\ldots,150$ are shown in Figure \ref{fig:ex:2:spectral_radius:actconfig:1} and
clearly confirms the closed-loop stability assessment. Respective
results are provided in Figure
\ref{fig:ex:2:spectral_radius:actconfig:2} for $\zeta=0.4$ and
Figure \ref{fig:ex:2:spectral_radius:actconfig:3} for
$\zeta=0.1$. The {values of $\rho_m$} computed for $m=150$ depending on the
order $n$ is depicted in Figure \ref{fig:ex:2:spectral_radius:nvar}
for the three actuator/sensor configurations. While closed-loop
stability is achieved in all scenarios only the collocated
configuration with $\zeta=\xi=0.4$ shows an (almost) uniform decay to
the assigned smallest closed-loop eigenvalue $\kappa_1 = -10$.
\begin{figure*}[!t]
  \centering
  \subcaptionbox{$\zeta=0.7$, $\epsilon=0.05$, $\xi=0.4$ \label{fig:ex:2:spectral_radius:actconfig:1}}
  {\includegraphics[width=0.32\textwidth]{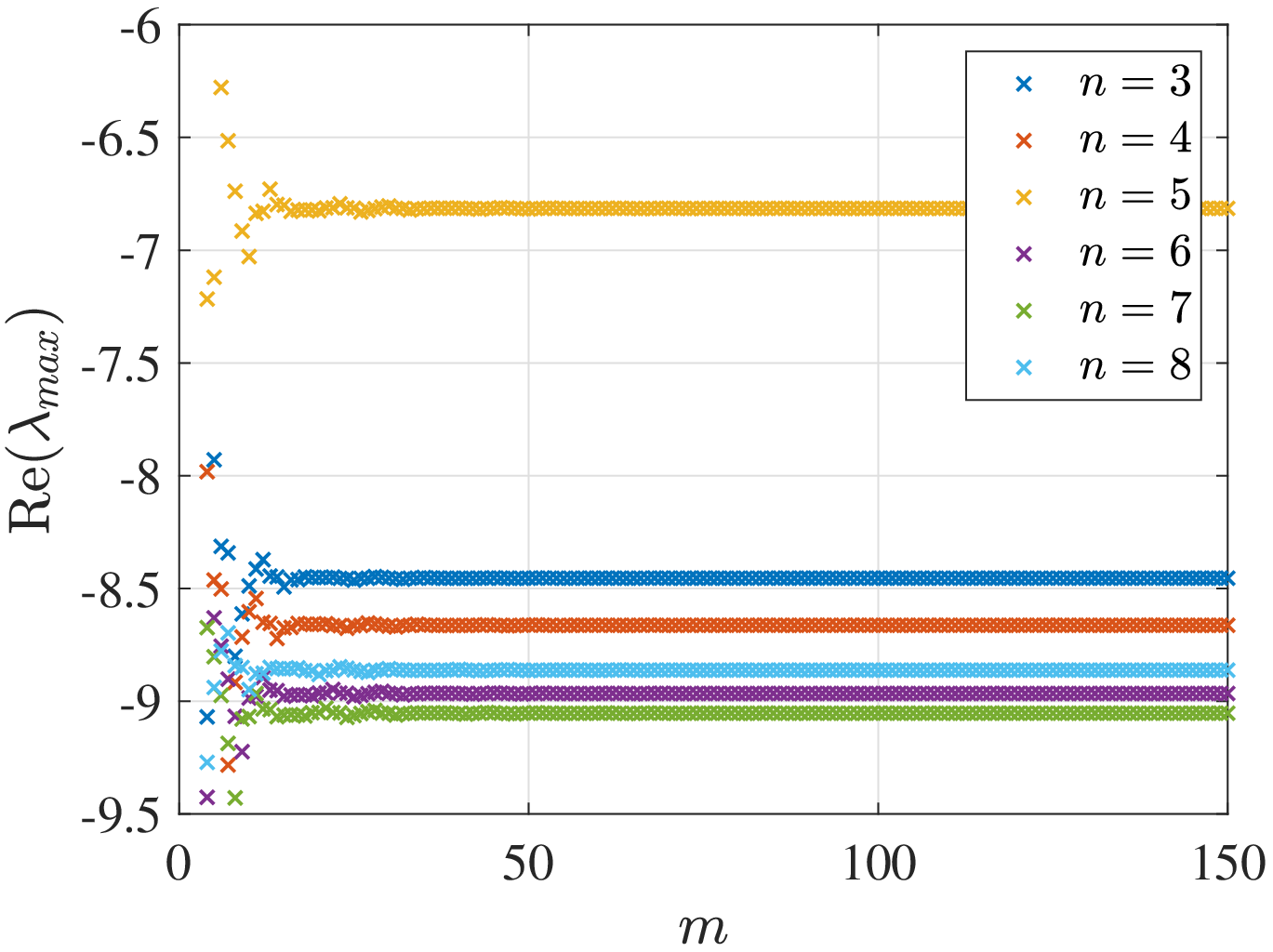}}
  \subcaptionbox{$\zeta=0.4$, $\epsilon=0.05$, $\xi=0.4$ \label{fig:ex:2:spectral_radius:actconfig:2}}
  {\includegraphics[width=0.32\textwidth]{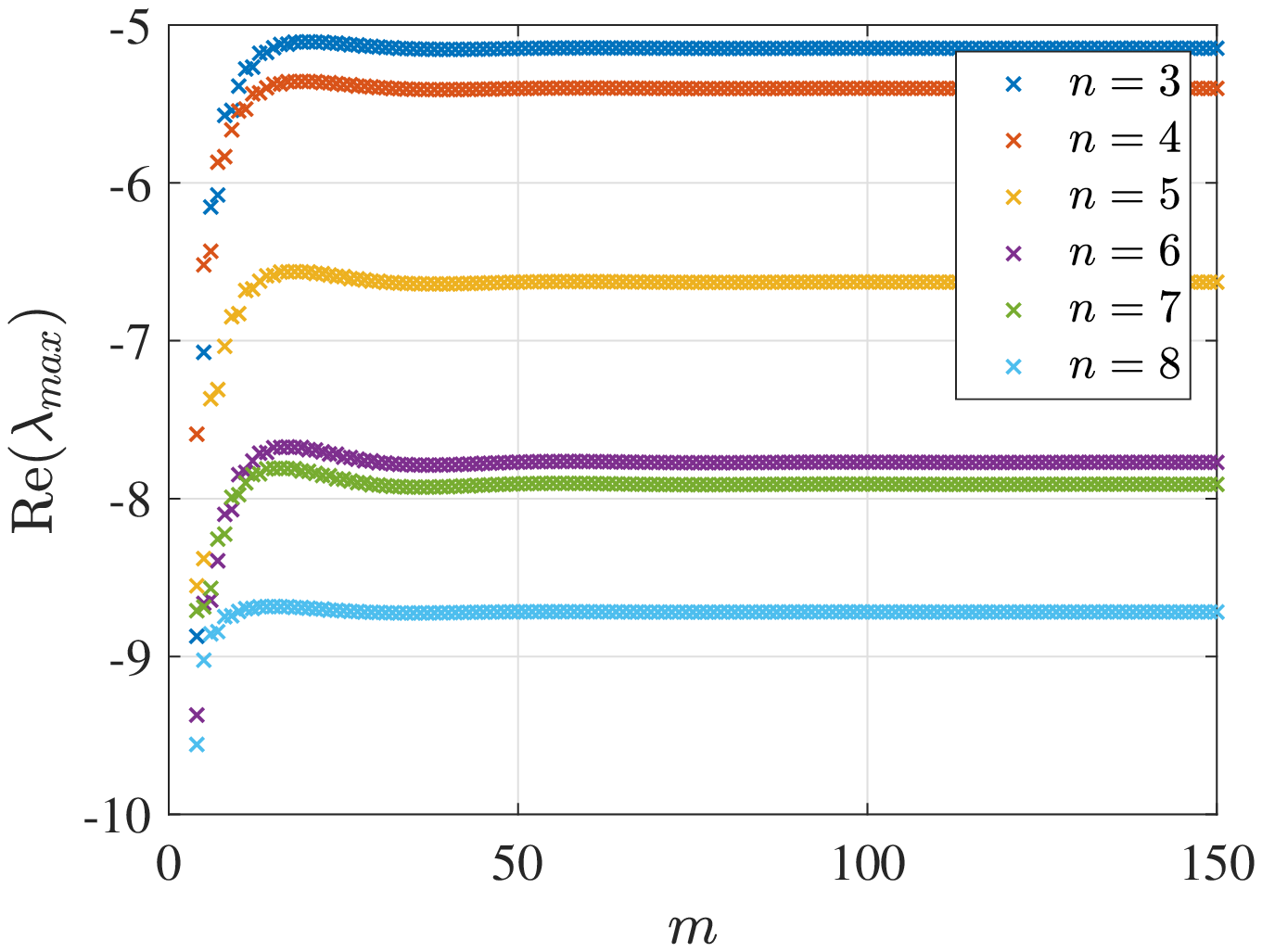}}\\
  \subcaptionbox{$\zeta=0.1$, $\epsilon=0.05$, $\xi=0.4$ \label{fig:ex:2:spectral_radius:actconfig:3}}
  {\includegraphics[width=0.32\textwidth]{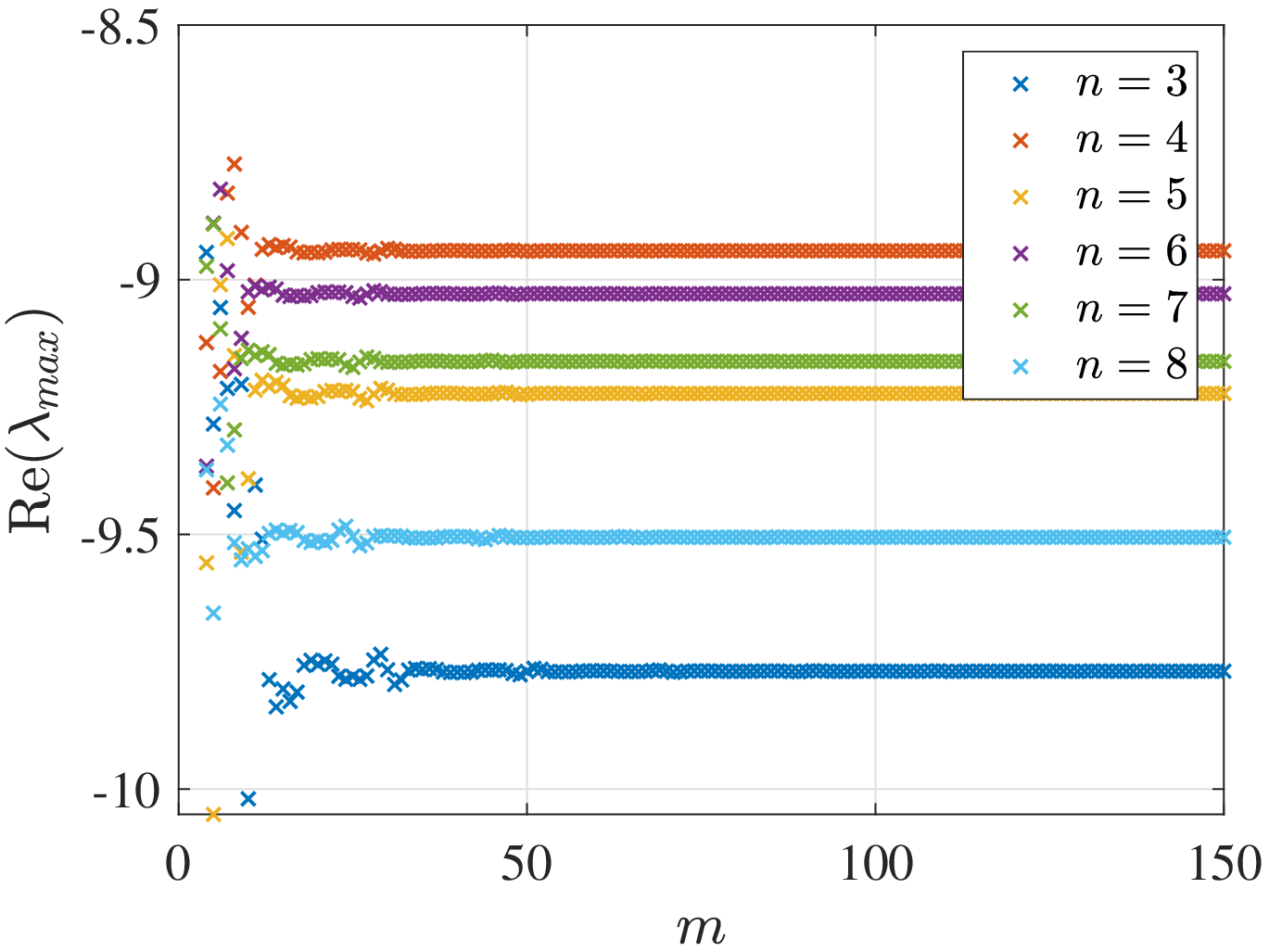}}
  \subcaptionbox{{Values of $\rho_m$} for $m=150$ and variation of the
    dimension $n$ of the slow subsystem for the configurations of
    figures \ref{fig:ex:2:spectral_radius:actconfig:1}-\ref{fig:ex:2:spectral_radius:actconfig:3}.\label{fig:ex:2:spectral_radius:nvar}}
  {\includegraphics[width=0.32\textwidth]{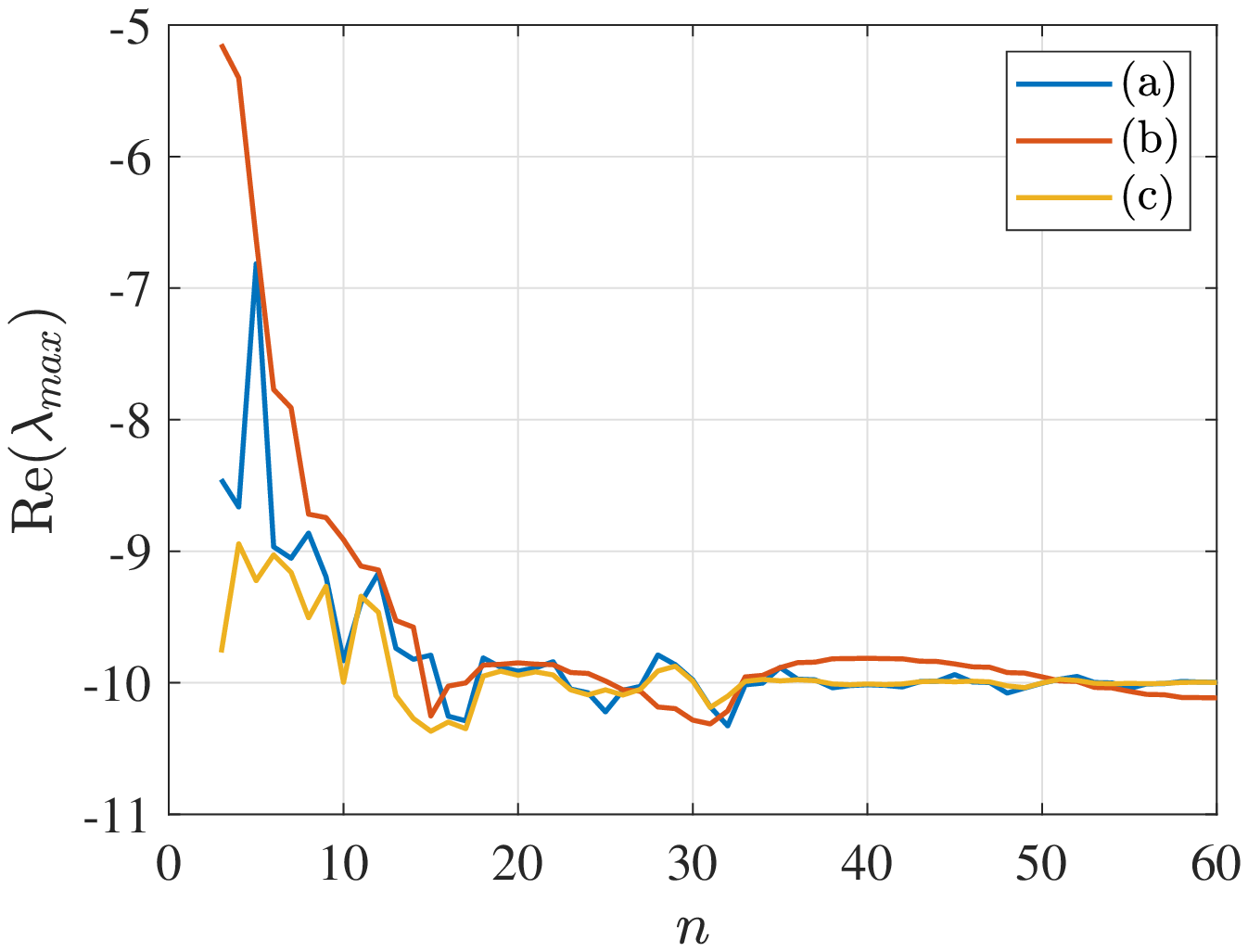}}
  \caption{Results for the scalar example of Section
    \ref{subsec:sim:ex2} with varying location of in-domain control and sensing for $n\in[3,8]$.}
  \label{fig:ex:2:spectral_radius}
\end{figure*}
\subsection{Coupled diffusion-reaction problem}\label{subsec:sim:ex3}
As a second problem consider
\begin{subequations}
  \label{eq:ex3:dr}
  \begin{align}
        &\partial_t \bs{x} = D\partial_z^2 \bs{x} + R\bs{x} + B\bs{u},&&z\in(0,1),~t>0\\
        &G_0\partial_z \bs{x}\vert_{z=0} + F_0 \bs{x}\vert_{z=0} = \bs{0},&&t>0 \\
        &G_1\partial_z \bs{x}\vert_{z=1} + F_1 \bs{x}\vert_{z=1} = \bs{0},&&t>0\\
        &\bs{x}\vert_{t=0}=\bs{x}_0,&&z\in[0,1]
      \end{align}
      for $N=2$ with the diffusion and reaction matrices 
      \begin{align}
        D=
        \begin{pmatrix}
          1 & 0\\0 & 2
        \end{pmatrix},
                     \quad
                     R=
                     \begin{pmatrix}
                       \alpha & r_{12}\\ r_{21} & \alpha
                     \end{pmatrix},
      \end{align}
      and
      \begin{align}
        G_0 = F_1 =\begin{pmatrix}
          1 & 0\\ 0 & 0
        \end{pmatrix},\quad
                      G_1 = F_0 = \begin{pmatrix}
                        0 & 0\\ 0& 1
                      \end{pmatrix}.
      \end{align}
      In-domain control is considered by
      means of
    \begin{align}
      B =
      \begin{pmatrix}
        f_{\zeta_1,\epsilon_1}\\
        f_{\zeta_2,\epsilon_2}
      \end{pmatrix}
    \end{align}
    with $f_{\zeta,\epsilon}$ as defined in
    \eqref{eq:ex2:rectangular_shape}. The output is taken as the
    pointwise measurement of the state, i.e.,
    \begin{align}
      y_1 = x_1\vert_{z=\xi_1},\quad y_2 = x_2\vert_{z=\xi_2}. 
    \end{align}
  \end{subequations}
  Problem can be recast into the abstract form on the state space
    $X=(L_2(0,1))^{2}$ by introducing
    \begin{align}
      A\bs{x} &= D\partial_z^2 \bs{x} + R\bs{x}\label{eq:ex3:opA}\\
      \intertext{with domain}
      D(A) &= \big\{ \bs{x}\in X \vert\, A\bs{x}\in X,\,
             G_0\partial_z \bs{x}\vert_{z=0} + F_0 \bs{x}\vert_{z=0}
             = \bs{0},\nonumber\\
      &\quad \, G_1\partial_z \bs{x}\vert_{z=1} + F_1
        \bs{x}\vert_{z=1} = \bs{0}\big\}\label{eq:ex3:domainA}
    \end{align}
For the verification of Assumption \ref{assump:general} it is
necessary to determine the eigenvalue distribution and the respective
eigenvectors from $A\bs{\phi}_k=\lambda_k\bs{\phi}_k$, $\bs{\phi}_k\in
D(A)$ with $A$ and $D(A)$ defined in \eqref{eq:dr:abstract}.
Contrary to the scalar example this is can no longer be
performed analytically. However, it is possible to compute asymptotic
results as $k\gg 1$. In particular, after some tedious but
straightforward computations the two asymptotic eigenvalue
branches $\lambda_{j,k} = \lambda_{j,k}^{a}+O(k^{-2})$, $j=1,2$, with

\begin{align}
  \label{eq:ex3:eigdist}
  \begin{split}
    {\lamAa} &= \alpha - {\frac{3\mu_k^2 -\sqrt{\mu_k^4+4
          r_{12} r_{21}}}{2}} \\
    {\lamBa} &= \alpha - {\frac{3\mu_k^2 + \sqrt{\mu_k^4+4
          r_{12} r_{21}}}{2}}
  \end{split}
\end{align}%
and $\mu_k=(2k-1)\pi/2$, can be deduced for $k\gg 0$. The corresponding eigenvectors
$\bs{\phi}_{1,k}$, $\bs{\phi}_{2,k}$ of  the operator $A$ and the mutually
orthogonal eigenvectors 
$\bs{\psi}_{1,k}$, $\bs{\psi}_{2,k}$ of the adjoint operator 
$A^\ast\bs{x} = D\partial_z^2\bs{x} + R^{T}\bs{x}$, $D(A^\ast)=D(A)$
can be determined by direct evaluation and scaled so that {$\langle
\bs{\phi}_{i,k},\bs{\psi}_{i,l} \rangle_X=\delta_{k,l}$} for
$i=1,2$. The analysis in \ref{sec:app} allows us to conclude
that the eigenvectors of $A$ generate a Riesz basis. As a consequence $A$ is a Riesz spectral operator 
\cite[{Section 2.3}]{curtain_zwart:95} so that using Fourier series expansion and projection
the system \eqref{eq:ex3:dr} can be re-written as the
infinite-dimensional system of ODEs in diagonal form
\begin{subequations}
  \begin{align}
    \dot{x}_k &= \lambda_k x_k + \bs{b}_k^T\bs{u},&&
                k\in\mathbb{N}\\
    x_k(0) &= \langle \bs{x}_0,\bs{\psi}_k\rangle_X = x_k^0
  \end{align}
\end{subequations}
with $x_k=\langle \bs{x},\bs{\psi}_k\rangle_X$ and $\bs{b}_k^T\bs{u} =
\langle B\bs{u},\bs{\psi}\rangle_X$. 
Herein, $\lambda_k$ is asymptotically determined by
\eqref{eq:ex3:eigdist} for $k$ sufficiently large. In any case for
finite $k$ the eigenvalues can be approximately computed by a suitable
discretization of \eqref{eq:ex3:dr}. 

Assuming that actuator and sensors determined by the parameter pairs
$(\zeta_1,\epsilon_1,\xi_1)$ and $(\zeta_2,\epsilon_2,\xi_2)$ are
chosen so that the slow finite-dimensional subsystem
\eqref{eq:findim:slow} is stabilizable and detectable, then Assumption
\ref{assump:general:a0} is fulfilled. Taking into account the
eigenvalue asymptotics \eqref{eq:ex3:eigdist} yields that
\ref{assump:general:a1} holds true. 
Assumption \ref{assump:general:a2a} follows from the fact that $B$ and thus the $\bs{b}_k$ are bounded and that $|\lambda_k|$ grows quadratically.
\begin{figure*}[!t]
  \centering
  \subcaptionbox{{Values of $\rho_m$} for $n\in[3,8]$ over $m$. \label{fig:ex:3:spectral_radius:1}}
  {\includegraphics[width=0.32\textwidth]{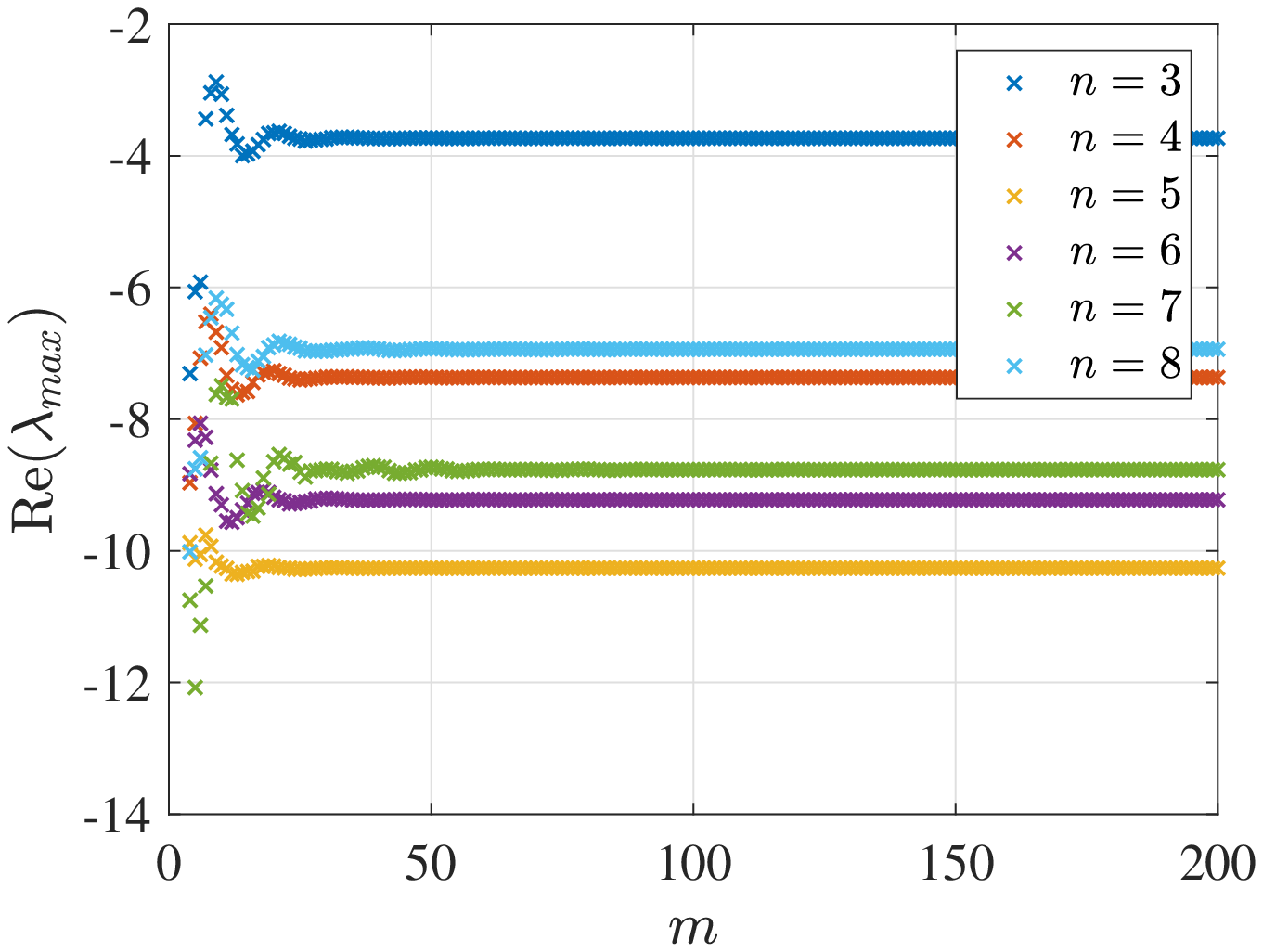}}
  \subcaptionbox{{Values of $\rho_m$} for $m=200$ over $n$. \label{fig:ex:3:spectral_radius:2}}
  {\includegraphics[width=0.32\textwidth]{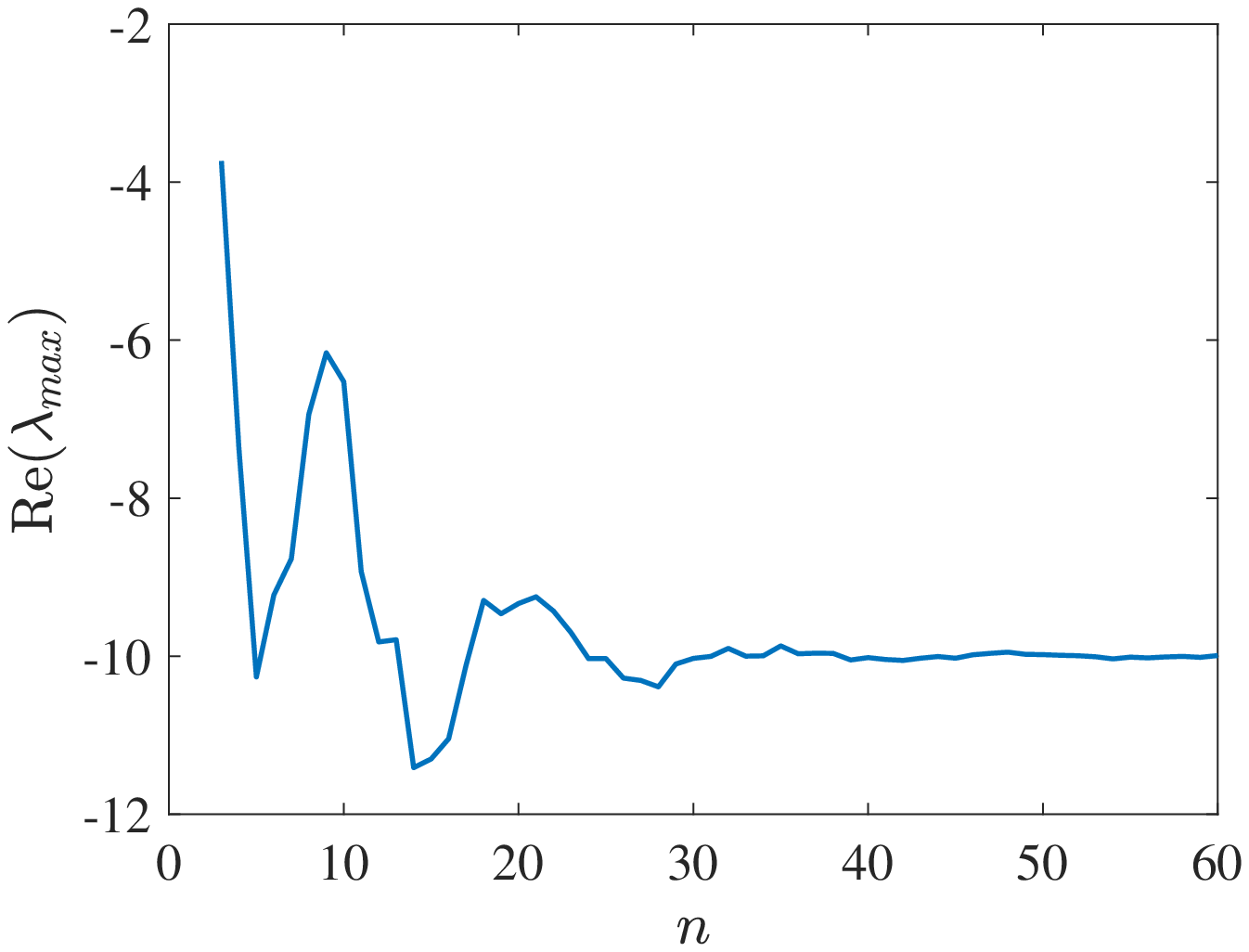}}
  \caption{Results for the coupled example of Section
    \ref{subsec:sim:ex3}.}
  \label{fig:ex:3:spectral_radius}
\end{figure*}

For numerical evaluation we consider $\alpha=10$, $r_{12}=5$,
$r_{21}=10$, which yields the open-loop eigenvalues $\lambda_1=11.56$,
$\lambda_2=2.84$, $\lambda_3=-13.09$, $\lambda_4=-33.79$,
$\lambda_k<\lambda_4$, $k\geq 5$ so that two eigenvalues have to be
shifted to the complex left half-plane by the feedback control. In
particular, in view of the discussed preliminaries, the stabilizing
feedback gain matrix $\bs{K}$ and the observer matrix $\bs{L}$ are determined
to place the two eigenvalues $\lambda_1$, $\lambda_2$ to the desired
eigenvalues $\kappa_1 = -10$, $\kappa_2 = -11$, $\kappa_k=\lambda_k$
for $k\in[3,n]$ for the feedback control and $\nu_1 = -15$, $\nu_2 =
-16$,  $\nu_k=\lambda_k$ for $k\in[3,n]$ for the observer error 
dynamics. Actuators and sensors are parametrized by $\zeta_1=0.3$,
$\epsilon_1=0.05$, $\xi_1=0$ and $\zeta_2=0.6$, $\epsilon_2=0.05$,
$\xi_2=1$ so that the outputs $y_1=x_1\vert_{z=0}$ and
$y_2=x_2\vert_{z=1}$ denote the boundary values of the state
variables. The resulting {values of $\rho_m$} for the dimension of the slow
subsystem restricted to $n\in[3,8]$ and $m=n+1\ldots,200$ residual
modes is shown in Figure \ref{fig:ex:3:spectral_radius:1}. The obtained
results clearly confirm the closed-loop stability assessment. Similar
to the previous examples the {values to which the $\rho_m$ converge} vary with $n$ and may grow as $n$ is increased. To further study this behavior
Figure \ref{fig:ex:3:spectral_radius:2} shows {the $\rho_m$} for
$m=200$ as they change over $n$. Here after some initial variation the
expected behavior becomes visible with the {values} settling to
the assigned smallest closed-loop eigenvalue $\kappa_1 = -10$. 
\begin{remark}
  As indicated before it is noteworthy to mention that the numerical
  results for all examples show that the desired dominating eigenvalue
  assigned during the state feedback control and state observer design
  (here $\kappa_1$) is obtained only for sufficiently large values of
  the order $n$ of the slow subsystem used for design. However, the
  results clearly indicate that closed-loop stability, as assessed in
  previous sections, is given for much lower values of $n$. This is
  an interesting observation that needs further examination. 
\end{remark}
\section{Conclusions}
The closed-loop stability of linear diffusion-reaction systems
under finite-dimensional observer-based state feedback control,i.e., dynamic output feedback control, is addressed based on the
classical decomposition of the considered class of infinite-dimensional
diffusion-reaction systems into a finite-dimensional slow subsystem
and an infinite-dimensional (residual) fast subsystem. State
feedback control and observer design is performed based on the slow
subsystem but remains interconnected to the residual system, which
leads to control and observation spillover. By thoroughly analyzing
the (dynamic) feedback interconnection of the subsystems a small-gain
theorem can be applied to verify closed-loop stability of the
infinite-dimensional system. For practical purposes an approach for
the computation of the required dimension of the slow subsystem used
for controller design is presented together with simulation results
scalar and coupled linear diffusion-reaction systems that confirm
the theoretical assessment.

\appendix
\section{Riesz basis generation for coupled diffusion-reaction problem}\label{sec:app}
To analyze the Riesz basis property of the set of eigenfunctions of
the operator $A$ for problem \eqref{eq:ex3:dr} it is necessary to take
into account the two eigenvalue branches provided in
\eqref{eq:ex3:eigdist} in terms of their asymptotics. These follow
from solving $A\bs{\phi}=\lambda\bs{\phi}$ with $\bs{\phi}\in D(A)$,
which after some tedious computations yields the characteristic
equation
\begin{multline}\label{eq:ex3:chareqn}
  \cos \left(\epsilon_{-}(\bar{\lambda})\right) \cos \left(\epsilon_{+}(\bar{\lambda})\right)
= -4r_{12} r_{21} \times \\
    \frac{(\bar{\lambda} ^2-r_{12} r_{21})+\frac{3}{\sqrt{8}} \bar{\lambda}  \sqrt{ \bar{\lambda} ^2- r_{12} r_{21}} \sin (\epsilon_{-}(\bar{\lambda}))
      \sin (\epsilon_{+}(\bar{\lambda}))}{(\bar{\lambda} ^2-r_{12} r_{21}) (\bar{\lambda} ^2+4
    r_{12} r_{21})}
\end{multline}
with $\bar{\lambda}=\lambda-\alpha$ and
\begin{align*}
  \epsilon_{\pm}(\bar{\lambda}) = \frac{1}{2} \sqrt{-3 \bar{\lambda} \pm\sqrt{\bar{\lambda} ^2+8
    r_{12} r_{21}}}.
\end{align*}
\subsection{Asymptotic analysis of the eigenvalues}
  To deduce \eqref{eq:ex3:eigdist} consider
  $\epsilon_{+}(\lamBar)=\mu\in\mathbb{R}^{+}$ and solve for $\lamBar$, i.e. 
  \begin{align}
    \label{eq:ex3:app:eigvals:branch:1}
    \lamBar = \lambda - \alpha =
    - \frac{3\mu^2}{2} + \frac{\sqrt{\mu^4+4 r_{12} r_{21}}}{2}.
  \end{align}
  This admits to conclude the following relationships
  \begin{enumerate}
  \item[(i)] $\sqrt{\lamBar^2+8r_{12}r_{21}} = 3\lamBar+4\mu^2\geq
    0$
  \item[(ii)] $\epsilon_{-}(\lamBar) =
    \frac{1}{2}\sqrt{-6\lamBar-4\mu^2}$, $-6\lamBar-4\mu^2\geq 0$
  \item[(iii)] $\epsilon_{-}(\lamBar) =
    \frac{1}{2}\sqrt{5\mu^2-3\sqrt{\mu^4+4r_{12}r_{21}}}\geq 0$ as
    $5\mu^2-3\sqrt{\mu^4+4r_{12}r_{21}}\geq 0$. 
  \end{enumerate}
  Using that for $\gamma\gg 1$ and fixed $c\in\R$ the inequality $\sqrt{\gamma^2 + c}= \gamma+O(\gamma^{-1})$ holds, for $\mu\gg 1$ we obtain
  \begin{enumerate}
  \item[(iv)] $\lamBar = -\mu^2 + {O(\mu^{-2})}$
  \item[(v)] $5\mu^2-3\sqrt{\mu^4+4r_{12}r_{21}} = 
    2\mu^2 + {O(\mu^{-2})}$\\[-2ex]
  \item[(vi)] $\lamBar^2-r_{12}r_{21} = \mu^4 + O(1)$
  \item[(vii)] $\lamBar^2+4r_{12}r_{21} = \mu^4 + O(1)$
  \end{enumerate}
  for each fixed (finite) value of $r_{12}r_{21}$. Property (v)
  together with (iii) implies {$\epsilon_{-}(\lamBar) = \mu+O(\mu^{-3})$ and $\epsilon_{+}(\lamBar) = \mu/\sqrt{2}+O(\mu^{-3})$} for $\mu\gg
  1$. Let $f(\lamBar)=\cos(\epsilon_{-}(\lamBar))
  \cos(\epsilon_{+}(\lamBar))$, let $g(\lamBar)$ denote
  the right hand side of \eqref{eq:ex3:chareqn}. In view of properties (iv) to
  (vii) we obtain for $\mu\gg 1$
  \begin{align*}
    f(\lamBar)
    &= \cos(\mu)\cos\bigg(\frac{\mu}{\sqrt{2}}\bigg) +
      {O(\mu^{-3})}\\
    g(\lamBar) &= 
                 {\frac{\mu^4+O(1) + (\mu^4+O(1)) \sin (\epsilon_{-}(\bar{\lambda}))
      \sin (\epsilon_{+}(\bar{\lambda}))}{(\mu^4+O(1))^2}}\\
    &= {O\left(\frac{\mu^4+O(1)}{(\mu^4+O(1))^2}\right) = O(\mu^{-4})}. 
  \end{align*}
  Hence $f(\lamBar)=g(\lamBar)$ implies
  \begin{align}
    \label{eq:app:chareqn:asympt:simp}
    \cos(\mu)\cos\bigg(\frac{\mu}{\sqrt{2}}\bigg) = {O(\mu^{-3})}.
  \end{align}
  A sequence of solutions for this equation is $\mu=\mu_k + O(k^{-3})$ 
  with $\mu_k=(2k-1)\pi/2$, $k\in\mathbb{N}$. The corresponding eigenvalue branch follows from
  the substitution into \eqref{eq:ex3:app:eigvals:branch:1} 
  \begin{multline}\label{eq:ex3:app:evals:branch:1:aympt}
    \lamBar_{1,k} = \lambda_{1,k} - \alpha = - \frac{3\mu^2}{2} +
    \frac{\sqrt{\mu^4+4 r_{12} r_{21}}}{2}\\
    =
    \lambAa + O(k^{-2}) 
  \end{multline}
  with
  \begin{align*}
    \lambAa = \lamAa - \alpha = 
    - \frac{3\mu_k^2}{2} + \frac{\sqrt{\mu_k^4+4 r_{12} r_{21}}}{2},
  \end{align*}
  where the second equality follows from (iv) and $\mu = \mu_k + O(k^{-3})$.
  The second sequence of asymptotic solutions to
  \eqref{eq:app:chareqn:asympt:simp}, $\mu=\bar{\mu}_k+O(k^{-3})$ with $\bar{\mu}_k=\sqrt{2}(2k-1)\pi/2$, corresponds to the analysis of
  the second branch determined from $\epsilon_{+}(\lamBar)=\mu\in\mathbb{R}$ so
  that following a similar argumentation the second eigenvalue branch
  can be determined in the form  
  \begin{multline}\label{eq:ex3:app:evals:branch:2:aympt}
    \lamBar_{2,k} = \lambda_{2,k} - \alpha =
    - \frac{3\mu^2}{2} - \frac{\sqrt{\mu^4+4 r_{12} r_{21}}}{2}\\
    = \lambBa + O(k^{-2})
  \end{multline}
  with
  \begin{align*}
    \lambBa = \lamBa - \alpha = - \frac{3\mu_k^2}{2} -
    \frac{\sqrt{\mu_k^4+4 r_{12} r_{21}}}{2}. 
  \end{align*}
  \subsection{Asymptotic analysis of the eigenvectors}
   Taking into account the two branches
  \eqref{eq:ex3:app:evals:branch:1:aympt} and
  \eqref{eq:ex3:app:evals:branch:2:aympt} the solution of the
  eigenproblem $A\bs{\phi}=\lambda\bs{\phi}$ with $\bs{\phi}\in D(A)$
  can be asymptotically determined. A closed-form general solution,
  which is determined up to a normalization constant, can be computed
  for each of the two eigenvalue branches. The normalization constant
  is obtained by evaluating $\|\bs{\phi}_{j,k}\|^2_X=\langle\bs{\phi}_{j,k},\bs{\psi}_{j,k}
  \rangle_X=1$ for the branches $j\in\{1,2\}$. Here, $\bs{\psi}_{j,k}$ denotes the
  eigenvector for the adjoint operator 
  $A^{*}$, which is given by $A^{*}\bs{x}=D\partial_z^2\bs{x}+R^T\bs{x}$ with
  $D(A^{*})=D(A)$. This implies that $\bs{\psi}_{j,k}$ follows
  from  $\bs{\phi}_{j,k}$ by mutually interchanging $r_{12}$ and
  $r_{21}$.
  
The resulting expressions are rather lengthy and are thus
  subsequently omitted. However, they allow to deduce the following
  asymptotics
  \begin{align}\label{eq:ex3:evecs:asym:general}  
    \bs{\phi}_{j,k} & = \bs{\phi}_{j,k}^{a} + O(k^{-1}),\quad j\in\{1,2\}
  \end{align}
  with
  \begin{align}\label{eq:ex3:evecs:asym:simplified:1}                   
    \bs{\phi}_{1,k}^{a} &=
      \frac{\sqrt{2}}{\sqrt{1+\frac{r_{12}r_{21}}{\mu_k^4}}}
      \begin{bmatrix}
        1\\[2ex]
        \frac{r_{21}}{\mu_k^2}
      \end{bmatrix} \cos(\mu_k z) 
  \end{align}
  for branch 1 with
  \eqref{eq:ex3:app:evals:branch:1:aympt} and
  \begin{align}\label{eq:ex3:evecs:asym:simplified:2}
    \bs{\phi}_{2,k}^{a}
    &=
      \frac{\sqrt{2}}{\sqrt{1+\frac{r_{12}r_{21}}{\mu_k^4}}}
      \begin{bmatrix}
        - \frac{r_{12}}{\mu_k^2}\\[2ex]
        1
      \end{bmatrix}
    \sin(\mu_k z)
  \end{align}
  for branch 2 with
  \eqref{eq:ex3:app:evals:branch:2:aympt}. To illustrate the
  asymptotic behavior, consider branch 1 with
  \eqref{eq:ex3:app:evals:branch:1:aympt},
  \eqref{eq:ex3:evecs:asym:simplified:2}, which yields 
  \begin{align*}
    A\bs{\phi}_{1,k}^{a}-\lamAa\bs{\phi}_{1,k}^{a} =
    \frac{\sqrt{2}\cos(\mu_k
    z)}{\sqrt{1+\frac{r_{12}r_{21}}{\mu_k^4}}}
    \begin{bmatrix}
      O(k^{-2})\\
      O(k^{-4})
    \end{bmatrix}
  \end{align*}
  with boundary conditions 
  \begin{align*}
    &\big(G_0\partial_z \bs{\phi}_{1,k}^{a} + F_0
    \bs{\phi}_{1,k}^{a}\big)\big\vert_{z=0} = \frac{\sqrt{2}r_{21}}{\sqrt{\mu_k^4+r_{12}r_{21}}}
      \begin{bmatrix}
        0\\1
      \end{bmatrix}
    =
    \begin{bmatrix}
        0\\O(k^{-2})
      \end{bmatrix}
\\
    &\big(G_1\partial_z  \bs{\phi}_{1,k}^{a} + F_1
    \bs{\phi}_{1,k}^{a}\big)\big\vert_{z=1} = 
      \frac{\sqrt{2}(-1)^kr_{21}\mu_k}{\sqrt{\mu_k^4+r_{12}r_{21}}}
      \begin{bmatrix}
        0\\1
      \end{bmatrix}
    =
    \begin{bmatrix}
        0\\O(k^{-1})
      \end{bmatrix}
  \end{align*}
  A similar analysis can be performed for branch 2 with
  \eqref{eq:ex3:app:evals:branch:2:aympt},
  \eqref{eq:ex3:evecs:asym:simplified:2}.   

  \subsection{Riesz basis property}
  To analyze that
$\{\bs{\phi}_k^{a}\}_{k\in\mathbb{N}}=\{\bs{\phi}_{1,k}^{a},\bs{\phi}_{2,k}^{a}\}_{k\in\mathbb{N}}$
generates a Riesz basis we make use of Bari's theorem
\cite{young:01,guo:ControlWaveBeam:2019} taking into account
\eqref{eq:ex3:evecs:asym:simplified:1}, \eqref{eq:ex3:evecs:asym:simplified:2} to
show that $\{\bs{\phi}_k^{a}\}_{k\geq  n} = \{\bs{\phi}_{1,k}^{a},\bs{\phi}_{2,k}^{a}\}_{k\geq
  n}$ for sufficiently large $n$ is
quadratically close to a (known) Riesz basis
$\{\bs{e}_k\}_{k\in\mathbb{N}}=\{\bs{e}_{1,k},\bs{e}_{2,k}\}_{k\in\mathbb{N}}$. For
the latter we consider the basis spanned by the eigenvectors of the
decoupled problem, i.e., \eqref{eq:ex3:opA}, \eqref{eq:ex3:domainA}
with matrix $R=0$. This implies 
\begin{align}
  \label{eq:app:known_riesz_basis}
  \bs{e}_{1,k} = \sqrt{2}
  \begin{bmatrix}
    \cos(\mu_k z)\\0
  \end{bmatrix},\quad
  \bs{e}_{2,k} = \sqrt{2}
  \begin{bmatrix}
    0\\
    \sin(\mu_k z)
  \end{bmatrix}
\end{align}
with $\mu_k=(2k-1)\pi/2$, $k\in\mathbb{N}$. As each of the sets 
$\{\sqrt{2}\cos(\mu_k z)\}_{k\in\mathbb{N}}$ and $\{\sqrt{2}\sin(\mu_k
z)\}_{k\in\mathbb{N}}$ generates a Riesz basis for $L_2(0,1)$ we conclude
that $\{\bs{e}_k\}_{k\in\mathbb{N}}$ defines a Riesz basis for
$X=(L_2(0,1))^2$. We remark that $\bs{e}_{1,k}$ and $\bs{e}_{2,k}$
refer to the eigenvalue branches $1$ and $2$ for the decoupled
problem as do $\bs{\phi}_{1,k}^{a}$ and $\bs{\phi}_{2,k}^{a}$ for the
considered coupled problem.

To verify that $\{\bs{\phi}_k^{a}\}_{k\geq  n}$ is quadratically close
to the Riesz basis $\{\bs{e}_k\}_{k\in\mathbb{N}}$ it is necessary to show that 
\begin{align}
  \label{eq:bari}
  \sum_{k\geq n}\| \bs{\phi}_k-\bs{e}_k  \|_X^2 < \infty 
\end{align}
for sufficiently large finite $n\in\mathbb{N}$. Taking into account
\eqref{eq:ex3:evecs:asym:general} provides
\begin{align*}
  \|\bs{\phi}_k-\bs{e}_k  \|_X^2 & = \| \bs{\phi}_{1,k}^a + O(k^{-1})
                                   - \bs{e}_{1,k}  \|_{X}^{2}\\
  &\qquad+
    \| \bs{\phi}_{2,k}^a + O(k^{-1}) - \bs{e}_{2,k}  \|_{X}^{2}
\end{align*}
with
\begin{multline*}
  \| \bs{\phi}_{j,k}^a + O(k^{-1})
  - \bs{e}_{j,k}  \|_{X}^{2}  =
  \| \bs{\phi}_{j,k}^a 
  - \bs{e}_{j,k}  \|_{X}^{2} + O(k^{-2}) \\
  +2 O(k^{-1})\int_{0}^{1} \big\{ (\phi^{a}_{j,k})_1 - (e_{j,k})_1 +
  (\phi^{a}_{j,k})_2 - (e_{j,k})_2 \big\} dz,
\end{multline*}
where $(\phi^{a}_{j,k})_l$, $(e_{j,k})_l$, $l\in\{1,2\}$ refer to the
$l$-th component of the vectors. 
Making use of \eqref{eq:ex3:evecs:asym:simplified:1}, \eqref{eq:ex3:evecs:asym:simplified:2} and
\eqref{eq:app:known_riesz_basis} we obtain
\begin{align*}
  \| \bs{\phi}_{1,k}^{a}-\bs{e}_{1,k}  \|_{X}^2
  &=
    \frac{r_{21}(r_{12}+r_{21})+2\mu_k^4\Big(1-\sqrt{1+\frac{r_{12}r_{21}}{\mu_k^4}}\Big)}{r_{12}r_{21}+\mu_k^4}\\
  &\qquad + O(k^{-2})\\
  \| \bs{\phi}_{2,k}^{a}-\bs{e}_{2,k}  \|_{X}^2
  &=
    \frac{r_{12}(r_{12}+r_{21})+2\mu_k^4\Big(1-\sqrt{1+\frac{r_{12}r_{21}}{\mu_k^4}}\Big)}{r_{12}r_{21}+\mu_k^4}
  \\ &\qquad + O(k^{-2})
\end{align*}
and
\begin{multline*}
  \int_{0}^{1} \big\{ (\phi^{a}_{j,k})_1 - (e_{j,k})_1 +
  (\phi^{a}_{j,k})_2 - (e_{j,k})_2 \big\} dz = \\
  \underbrace{\frac{\sqrt{2}(1-(-1)^{k})}{\mu_k}}_{=(\star)}\underbrace{\Bigg(\frac{1}{\sqrt{1+\frac{r_{12}r_{21}}{\mu_k^4}}}-1\Bigg)}_{=(\star\star)}
  + \underbrace{\frac{r_{21}-r_{12}}{\mu_k\sqrt{\mu_k^4+r_{12}r_{21}}}}_{=(\star\star\star)}.
\end{multline*}
As $(\star)=O(k^{-1})$ and $(\star\star)=o(1)$ the product fulfills
$(\star)(\star\star)=o(k^{-1})$. Since $(\star\star\star)=O(k^{-3})$
we obtain
\begin{multline*}
  2 O(k^{-1})\int_{0}^{1} \big\{ (\phi^{a}_{j,k})_1 - (e_{j,k})_1 +
  (\phi^{a}_{j,k})_2 - (e_{j,k})_2 \big\} dz \\
   = O(k^{-1}){(o(k^{-1}) + O(k^{-3}))} = o(k^{-2}) + O(k^{-4}). 
\end{multline*}
Moreover observing $1-\sqrt{1+{r_{12}r_{21}}/{\mu_k^4}}  = O(\mu_k^{-4})$ for
$k\gg 1$, it follows that
\begin{align*}
  \| \bs{\phi}_{1,k}^{a}-\bs{e}_{1,k}  \|_{X}^2
  &=
    {\frac{r_{21}(r_{12}+r_{21})+ 2 \mu_k^4O(\mu_k^{-4})}{r_{12}r_{21}+\mu_k^4}
  = O(k^{-4})}\\
  \| \bs{\phi}_{2,k}^{a}-\bs{e}_{2,k}  \|_{X}^2
  &=
    {\frac{r_{12}(r_{12}+r_{21})+ 2 \mu_k^4O(\mu_k^{-4})}{r_{12}r_{21}+\mu_k^4} = O(k^{-4})}. 
\end{align*}
Thus, $\| \bs{\phi}_k^{a}-\bs{e}_k \|_X^2=O(k^{-2})$ and consequently $\sum_{k\geq n}\| \bs{\phi}_k^{a}-\bs{e}_k \|_X^2$ converges and
\eqref{eq:bari} is fulfilled, so that
$\{\bs{\phi}_k\}_{k\in\mathbb{N}}=\{\bs{\phi}_{1,k}^{a},\bs{\phi}_{2,k}^{a}\}_{k\in\mathbb{N}}$
generates a Riesz basis.

\bibliographystyle{plain}
\bibliography{refs}

\end{document}